\documentclass{techrep}

\usepackage{amsfonts}
\usepackage{amsthm}

\newcommand{\comment}[1]{}

\newtheorem{theorem}{Theorem}{\bfseries}{\itshape}

\newtheorem{lemma}{Lemma}{\bfseries}{\itshape}

{\bfseries}{\itshape}

{\bfseries}{\itshape}

{\bfseries}{\itshape}

{\bfseries}{\itshape}

\newtheorem{assumption}{Assumption}{\bfseries}{\rm}

\theoremstyle{definition}

\newtheorem{definition}{Definition}

{\bfseries}{\rm}

\theoremstyle{plain}

\usepackage{verbatim}
\usepackage{algorithm2e}
\usepackage[utf8]{inputenc}
\usepackage{microtype}
\usepackage{amsthm,amsmath,graphicx}
\usepackage{caption}
\usepackage{subcaption}
\usepackage{hyperref}
\usepackage[table,dvipsnames]{xcolor}
\definecolor{linkblue}{named}{Blue}
\hypersetup{colorlinks=true, linkcolor=linkblue,  anchorcolor=linkblue,
	citecolor=linkblue, filecolor=linkblue, menucolor=linkblue,
	urlcolor=linkblue} 
\usepackage{wasysym}
\usepackage{graphicx}
\usepackage{caption}
\usepackage{enumitem}
\usepackage{thmtools, thm-restate}
\usepackage{wrapfig}
\usepackage{float}
\usepackage{stackrel}
\usepackage{fixfoot}
\usepackage[utf8]{inputenc}
\usepackage[english]{babel}
\usepackage{mathrsfs}  
\usepackage{venturis}
\usepackage{placeins}

\newlength\problemsep
\newcommand{\jit}[1]{}

\title{The Exact Spanning Ratio of the Parallelogram Delaunay Graph. \thanks{This research is supported in part by NSERC.}}

\author{Prosenjit Bose, Jean-Lou De Carufel, Sandrine Njoo}

\begin{document}

\date{\today}

\maketitle
\begin{abstract}
Finding the exact spanning ratio of a Delaunay graph has been one of the longstanding open problems in Computational Geometry. Currently there are only four convex shapes for which the exact spanning ratio of their Delaunay graph is known: the equilateral triangle, the square, the regular hexagon and the rectangle. In this paper, we show the exact spanning ratio of the parallelogram Delaunay graph, making the parallelogram the fifth convex shape for which an exact bound is known. The worst-case spanning ratio is \emph{exactly} 
$$\frac{\sqrt{2}\sqrt{1+A^2+2A\cos(\theta_0)+(A+\cos(\theta_0))\sqrt{1+A^2+2A\cos(\theta_0)}}}{\sin(\theta_0)},$$
where $A$ is the aspect ratio and $\theta_0$ is the non-obtuse angle of the parallelogram. Moreover, we show how to construct a parallelogram Delaunay graph whose spanning ratio matches the above mentioned spanning ratio. 
\end{abstract}

\section{Introduction}
In computer science, many phenomena can be represented by graphs, such as computer networks. Often, these graphs have many edges, and computations can be simplified if fewer edges are retained, for example, in compact routing, broadcasting etc. Throwing away edges leaves a subgraph, and if it has certain distance preserving properties, then it is called a spanner. As described by Narasimhan and Smid~\cite{DBLP:books/daglib/0017763}, the goal is to have a \emph{good spanner}, but what constitutes a good spanner depends on the application. In some cases, one may want a minimum weight spanning tree, whereas in other cases, it might be allowable to form cycles in order to keep all the points more connected, as long as the number of edges remains sufficiently small. Even more, some spanners are designed to be fault-tolerant, being able to lose edges without drastically affecting their ability to preserve distance. 

A \emph{geometric graph} $G$ is an undirected and weighted graph whose vertices and edges are points and line segments in the plane, where the weight of an edge is the Euclidean distance (denoted by $d_2(\cdot,\cdot)$) between its two endpoints.
When a geometric graph has certain distance preserving properties, it is called a \emph{spanner}.
More specifically,
given a point set $\mathcal{P}$,
a geometric graph with vertex set $\mathcal{P}$ is a \emph{$c$-spanner} if, for any $a,b\in\mathcal{P}$, there is a path in $G$ between $a$ and $b$, whose length is less than
or equal to $c\cdot d_2(a,b)$~\cite{DBLP:books/daglib/0017763}. 
The smallest constant $c$ for which this is true is called the \emph{spanning ratio} of $G$.
A simple example of a $1$-spanner is the complete graph. However, in the case of a complete graph, the number of edges is quadratic in the number of vertices. As such, using a subgraph with a linear number of edges is thus desirable.

We can also study the spanning ratio of a family $\mathcal{F}$ of geometric graphs.
In this case, we say that $c$ is an upper bound on the spanning ratio of $\mathcal{F}$ if all graphs in $\mathcal{F}$ are $c$-spanners.
Moreover, we say that $c'$ is a lower bound
on the spanning ratio of $\mathcal{F}$ if there is at least one graph in $\mathcal{F}$ whose spanning ratio is $c'$.
If we find matching upper and lower bounds on the spanning ratio of $\mathcal{F}$, we say that we have an exact (or tight) spanning ratio for this family.

For example, Delaunay triangulations, $\Theta_k$-graphs and Yao graphs are spanners with a linear number of edges~\cite{DBLP:books/daglib/0017763} (refer to Section~\ref{sect.known.results}). 
This paper will focus specifically on the spanning ratio of Delaunay graphs. Delaunay graphs are a fundamental structure that have been intensely
investigated in Computational Geometry~\cite{DBLP:books/wi/OkabeBSCK00}. There is an edge between two points
$a$ and $b$ in the Delaunay graph provided there exists a circle with $a$ and
$b$ on its boundary and no other points in its interior.
Apart from being spanners, Delaunay graphs possess many interesting geometric
properties. For example, among all triangulations of a set of points, the Delaunay graph is the one that maximizes the minimum angle~\cite{DBLP:books/wi/OkabeBSCK00}.

Determining the exact spanning ratio of the Delaunay graph is a notoriously
difficult problem.  Dobkin et al.~\cite{DBLP:journals/dcg/DobkinFS90} were the
first to show that the Delaunay graph has a spanning ratio of at most
$\pi(1+\sqrt{5})/2\approx 5.08$.  This was later improved by Keil and
Gutwin~\cite{DBLP:journals/dcg/KeilG92} who showed an upper bound of
$4\pi/3\sqrt{3}\approx 2.42$.  Currently, the best known upper bound of $1.998$
was shown by Xia~\cite{DBLP:journals/siamcomp/Xia13}. 
Chew~\cite{DBLP:journals/jcss/Chew89} gave a lower bound of $\pi/2$ which was
long believed to be optimal until Bose et al.~\cite{DBLP:journals/comgeo/BoseDLSV11} proved a lower bound of $1.5846$. This was later improved to $1.5932$ by  Xia and
Zhang~\cite{DBLP:conf/cccg/XiaZ11}. A tight bound on the
spanning ratio of the Delaunay graph remains elusive.

Several variations of the Delaunay triangulation come from generalizing
the circle to other convex shapes, Bose et al.~\cite{DBLP:journals/jocg/BoseCCS10} showed that the Delaunay graph defined by the homothet of any convex shape $C$ has a spanning ratio that is bounded by a constant times the ratio of the perimeter of $C$ to its width. Intuitively, this suggests that when the empty region is a long and skinny shape, the spanning ratio is large. Moreover, Bose et al.~\cite{DBLP:journals/cagd/BoseCS21} then showed that a Delaunay graph defined by any affine transformation $C'$ of $C$ is a constant spanner where the spanning ratio depends on the eigenvalues of the affine transformation. However, even if the spanning ratio for shape $C$ is tight, the upper bound obtained on the spanning ratio for $C'$ is not necessarily tight. In the search for tight bounds on the spanning ratio of Delaunay graphs, until recently, exact bounds were known for only 3 shapes, namely equilateral triangles~\cite{DBLP:journals/jcss/Chew89}, squares~\cite{DBLP:journals/comgeo/BonichonGHP15} and regular hexagons~\cite{DBLP:journals/jocg/Perkovic0T21}. The spanning ratio of empty rectangle Delaunay graphs was first studied by Bose et al.~\cite{DBLP:journals/comgeo/BoseCR18} who showed a spanning ratio of $\sqrt{2}(2A+1)$, where $A$ is the aspect ratio of the rectangle. This corroborates the intuition that long skinny rectangles have large spanning ratio. Recently,
van Renssen et al.~\cite{DBLP:conf/esa/RenssenSSW23} found a way to generalize Bonichon et al.'s result to give a tight bound of $\sqrt{2}\sqrt{A^2+1+A\sqrt{1+A^2}}$ for empty rectangle Delaunay graphs. This is the fourth shape for which we now have an exact bound on the spanning ratio. Their result is a delicate case analysis where they study different cases depending on the aspect ratio of the empty rectangle and the type of edge that results from these empty rectangles. Our main result is that we push the envelope further by proving a tight bound on the spanning ratio of Delaunay graphs defined by empty parallelograms. We generalize ideas from Bonichon et al.~\cite{DBLP:journals/comgeo/BonichonGHP15}, van Renssen et al.~\cite{DBLP:conf/esa/RenssenSSW23} and Bose et al.~\cite{DBLP:journals/cagd/BoseCS21}. The key idea was to find a way to change basis vectors without applying an affine transformation as in \cite{DBLP:journals/cagd/BoseCS21}. The generalization can be seen as follows: in Bonichon et al.'s approach, there is no degree of freedom in the shape defining the Delaunay graph since the aspect ratio of the square is fixed. In van Renssen et al.'s case, the shape has one degree of freedom, namely the aspect ratio. In our setting, the difficulty that arises is that our shapes have two degrees of freedom: the aspect ratio and the angle between adjacent sides of the parallelogram. We obtain an exact worst-case bound, and in fact, if we set the angle to $\pi/2$, we obtain van Renssen et al.'s result and in addition if we set the aspect ratio to $1$, we obtain Bonichon et al.'s result.

\subsection{Preliminaries, Notations and Definitions}
\label{sect.prelim}

The Delaunay triangulation of a point set $\mathcal{P}$ is a geometric graph $DT(\mathcal{P})=(\mathcal{P},E)$ defined as follows.
There is an edge between two vertices $u,v\in\mathcal{P}$ if and only if there exists a circle, with $u$ and $v$ on the boundary, which does not contain any point of $\mathcal{P}$ in its interior. Equivalently, there is a triangle $\triangle uvw$ if there is a circle with $u,v,w$ on its boundary which does not contain any point of $\mathcal{P}$ in its interior. 
To avoid degeneracies, we usually assume that the point set $\mathcal{P}$ is in \emph{general position}. Specifically, we assume that no three points lie on a common line and no four points lie on the boundary of a common circle.

Variations of the Delaunay triangulation exist, for example when the convex shape used is no longer a circle. These variations have been shown to be spanners as well~\cite{DBLP:conf/soda/BoseCHS19},\cite{DBLP:conf/wg/BonichonGHI10},\cite{DBLP:journals/jocg/Perkovic0T21},\cite{DBLP:journals/jcss/Chew89},\cite{DBLP:conf/esa/RenssenSSW23}, etc. In this paper, we consider the Delaunay graph where the convex shape is a parallelogram. 

Consider an arbitrary parallelogram $P$.
Let us denote the two side lengths by $\ell$ and $s$, respectively, where $\ell \geq s > 0$. We refer to the side with length $\ell$ as the \emph{long side}, and to the side with length $s$ as the \emph{short side} (even though the case where $\ell = s$ is allowed).
The \emph{parallelogram Delaunay graph} is a geometric graph with $\mathcal{P}$ as the vertex set. Given two vertices $a,b\in\mathcal{P}$, there is an edge between $a$ and $b$ if and only if there exists a scaled translate of $P$ with $a$ and $b$ on its boundary which contains no vertices of $\mathcal{P}$ in its interior.
Observe that different parallelograms give different Delaunay graphs.
Moreover, we emphasize the fact that rotations are \emph{not} allowed, only scaled translate of $P$.

Without loss of generality, we assume that the long side is vertical and the short side has non negative slope. Let $\theta_0$ be the non obtuse angle between the long and the short side of $P$. Note that any other orientation of the parallelogram can be seen as an isometry of the point set.  
The general position assumption we make on the point set $\mathcal{P}$ in this context is the following. We assume that no four vertices lie on the boundary of any scaled translate parallelogram of $P$ and that no two vertices lie on a line parallel to the sides of $P$. This means that no two vertices lie on a vertical or a horizontal line. 
In general parallelogram Delaunay graphs are near-triangulations. A \emph{near-triangulation} is a planar graph such that every bounded face is a triangle. As such, we denote them by $T_{\mathcal{P}}$ or simply $T$ when the point set is clear from the context.

\subsection{Known Results}
\label{sect.known.results}

Currently, a tight bound is not know for the classical Delaunay triangulation. The best known upper bound is $1.998$~\cite{DBLP:journals/siamcomp/Xia13}, while the best known lower bound is $1.5932$~\cite{DBLP:conf/cccg/XiaZ11}. Variations of the Delaunay graph for which we know the tight spanning ratio are limited; there are only four convex shapes for which we know the exact spanning ratio of their Delaunay graph: the equilateral triangle, the square, the regular hexagon and the rectangle. When the equilateral triangle is used, this is often referred to as the TD Delaunay, for triangular distance, in literature. Chew~\cite{DBLP:journals/jcss/Chew89} showed that in this case, the spanning ratio is exactly $2$. When the convex shape is a square, Chew showed that the spanning ratio is at most $\sqrt{10}$.
After more than $25$ years, Bonichon et al.~\cite{DBLP:journals/comgeo/BonichonGHP15}
had a breakthrough proving the exact spanning ratio of $\sqrt{4+2\sqrt{2}} \approx 2.61$ for the case where the convex shape is a square.
This motivated researchers to search for the exact spanning ratio with respect to other shapes.
Perkovic et al.~\cite{DBLP:journals/jocg/Perkovic0T21} showed the spanning ratio is exactly $2$ for hexagons. Inspired by Bonichon et al.~\cite{DBLP:journals/comgeo/BonichonGHP15}
van Renssen et al.~\cite{DBLP:conf/esa/RenssenSSW23} showed that the spanning ratio is exactly $\sqrt{2}\sqrt{A^2+1+A\sqrt{1+A^2}}$, for rectangles where $A$ is the aspect ratio. 

Another example of geometric graphs defined on a point set $\mathcal{P}$ that are spanners are $\Theta_k$-graphs. $\Theta_k$-graphs are defined for any integer $k\geq 3$. First, for each $i$ (where $0\leq i <k$), let $\mathcal{R}_i$ be the ray emanating from the origin that forms an angle of $\frac{2i\pi}{k}$ with the negative $y$ axis (by convention $\mathcal{R}_k=\mathcal{R}_0$). For a point $v\in\mathcal{P}$ and an index $i$ (where $0\leq i <k$), let $R_i^v$ be the ray emanating from $v$ that is parallel to $R_i$. Also define $C_i^v$ to be the cone consisting of all the points in the plane that are strictly between $R_i^v$ and $R_{i+1}^v$ or on $R_{i+1}^v$. Then the $\Theta_k$-graph of a point set $\mathcal{P}$ is the graph that has an edge $(v,w_i)$ if vertex $w_i$ lies in the cone $C_i^v$ and the perpendicular projection of $w_i$ onto the bisector of $C_i^v$ is the closest to $v$ compared to that of all other points in $(\mathcal{P}\backslash\{v\})\cap C_i^v$. Note that any vertex has at most $k$ outgoing edges. For example, the best known upper bound for the spanning ratio of the $\Theta_4$-graph~\cite{DBLP:conf/soda/BoseCHS19} was shown to be $17$. In the case of the $\Theta_5$-graph, the current best known upper bound has been shown to be $\frac{\sin(3\pi/10)}{\sin(3\pi/5)-\sin(3\pi/10)}<5.70$~\cite{DBLP:conf/wads/BoseHO21}. For the $\Theta_6$-graph, Bonichon et al.~\cite{DBLP:conf/wg/BonichonGHI10} showed that the $\Theta_6$-graph is the union of two spanning TD-Delaunay graphs and has a tight spanning ratio of $2$. Moreover, Bose et al.~\cite{DBLP:journals/tcs/BoseCMRV16} showed a tight bound of $1+2\sin(\pi/k)$ for $\Theta_k$-graphs, where the number of cones is $k=4m+2$, for $m\geq 1$. Next, they also show that when $k=4m+4$ for $m\geq 1$, the $\Theta_k$-graph has spanning ratio at most $\frac{1+2\sin(\pi/k)}{\cos(\pi/k)-sin(\pi/k)}$ and at least $1+2\tan(\pi/k)+2\tan^2(\pi/k)$ and for $\Theta_k$-graphs with $k=4m+3$ or with $k=4m+5$ cones, the spanning ratio is at most $\frac{\cos(\pi/2k)}{\cos(\pi/k)-\sin(3\pi/2k)}$. 

Yao graphs are defined in a similar way as the $\Theta_k$-graphs, where we consider cones around each vertex. Formally, let $Y_k$ denote the Yao graph with fixed integer $k>0$.  The point set is then partitioned such that around each vertex we have $k$ equiangular cones of angle $2\pi/k$. Each vertex is then connected to its closest neighbour in each cone. For all integer values of $k$, it is known whether or not  $Y_k$ is a geometric spanner~\cite{DBLP:journals/jocg/BarbaBDFKORTVX15}. For example, Barba et al.~\cite{DBLP:journals/jocg/BarbaBDFKORTVX15}, showed an upper bound on the spanning ratio for odd $k\geq 5$ of $1/(1-2\sin(\frac{3\cdot 2\pi}{8k}))$, for the case where $k=5$, Barba et al. further improved the upper bound, thus showing that $Y_5$ has a spanning ratio of at most $2+\sqrt{3} \approx 3.74$.  Next, the authors also show that in the case where $k=6$, $Y_k$ has a spanning ratio of at most $5.8$. Barba et al., also give a lower bound of $2.87$ for $Y_5$. While some bounds are known for the spanning ratio of Yao graphs, there are currently no known tight bound on their spanning ratios. 

\subsection{Our Contributions}

This paper further generalizes and extends the proof given by Bonichon et al.~\cite{DBLP:journals/comgeo/BonichonGHP15}, and van Renssen et al.~\cite{DBLP:conf/esa/RenssenSSW23}.
All together, we get an exact bound 
of
$$\frac{\sqrt{2}\sqrt{1+A^2+2A\cos(\theta_0)+(A+\cos(\theta_0))\sqrt{1+A^2+2A\cos(\theta_0)}}}{\sin(\theta_0)},$$
on the spanning ratio of the parallelogram Delaunay graph,
where $A$ is the aspect ratio and $\theta_0$ is the non-obtuse angle of the parallelogram.

We summarize our proof technique here.
To obtain an upper bound on the spanning ratio, we proceed by induction on the rank of the distances between pairs of points. To go from a start point $a$ to a destination point $b$, we only consider the paths that use endpoints of segments that cross the line segment $ab$. If a point is above (respectively below) the line $ab$, we refer to that point as a high (resp. low) point. Moreover, we only consider the case where the start point $a$ is on the bottom left corner and the destination point $b$ is on the rightmost side of the parallelogram.  

To bound the length of the path from $a$ to $b$, we break the path into three subpaths. The first part is bounded using the so-called \emph{Crossing Lemma} (see below). The second part considers the length of the subpath where the points all have $y$-coordinates that are significantly greater than that of $b$. On this path, the points are either all high points or all low points. The third part of the path is bounded using the induction hypothesis. 

The aforementioned Crossing Lemma is a crucial part in the proof.  The goal of the Crossing Lemma is to get a bound on the length of the path from $a$ until we encounter a parallelogram that has an edge between a high and a low point that is not \emph{too steep}. Using an inductive argument, we use techniques such as bounding monotone paths using the $L_1$-norm as well as bounding edge lengths in terms of their horizontal component. 

\section{Main Result}
\subsection{The Upper Bound}
Consider an arbitrary parallelogram $P$.
Let us denote the two side lengths by $\ell$ and $s$, respectively, where $\ell \geq s > 0$. We refer to the side with length $\ell$ as the \emph{long side}, and to the side with length $s$ as the \emph{short side} (even though the case $\ell = s$ is allowed). 
Without loss of generality, we assume that the long side is vertical and the short side has non-negative slope, let $\theta_0$ be the non obtuse angle between the long and the short side of $P$.
%We will naturally refer to the four sides of a parallelogram as W, N, E, S.
We define the aspect ratio of $P$, denoted by $A$, as $A:=\frac{\ell}{s}$.  For a point $a$ in the plane, let $x_a$ and $y_a$ be the $x$- and $y$-coordinates of $a$, respectively. Since the parallelogram Delaunay graph of a point set is a near-triangulation, we denote it as $T$.

In this section, we prove that the worst-case spanning ratio of $T$ is at most $h(A, \theta_0):=$
$$\frac{\sqrt{2}\sqrt{1+A^2+2A\cos(\theta_0)+(A+\cos(\theta_0))\sqrt{1+A^2+2A\cos(\theta_0)}}}{\sin(\theta_0)}.$$
We obtain this result after a few preparatory lemmas and observations.

Let $a,b$ be any two vertices in $T$. Without loss of generality, assume $a=(0,0)$ and $x_b > 0$. Let $d_2^T(a,b)$ be the length of the shortest path in $T$ between $a$ and $b$. 
We prove that $d_2^T(a,b)$ is at most $h(A,\theta_0)\,d_2(a,b)$ with equality occurring in the worst case, where $d_2(a,b)$ is the Euclidean distance from $a$ to $b$.
We denote the slope of the segment $ab$ as $S:=\frac{y_b}{x_b}$. Our proof considers four different scenarios, based on the value of $S$.
Each scenario is illustrated in Figure~\ref{fig:para2}. 
\begin{figure}
    %\centering
   \captionsetup{justification=centering}
   \centerline{
    \includegraphics[scale=1]{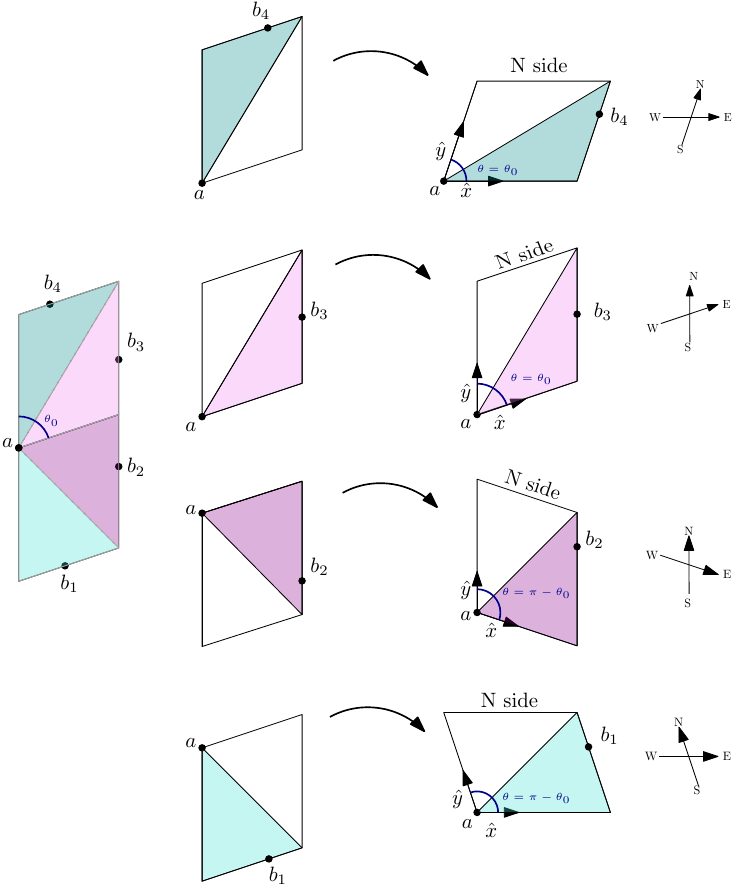}}
    \caption{The four scenarios considered about the slope between $a$ and $b$.}
    \label{fig:para2}
\end{figure}

In each scenario, we write the coordinates of all points with respect to a new basis\footnote{We do not perform a linear transformation. We simply write some objects with respect to the $\{\hat{x},\hat{y}\}$ basis instead of the usual basis. All distances stay the same.} $\{\hat{x},\hat{y}\}$.
We will denote the counterclockwise angle between $\hat{x}$ and $\hat{y}$ by $\theta$. Using this new basis, $P$ will be seen as a rectangle.
We denote by $\hat{P}$ the parallelogram $P$ written in the $\{\hat{x},\hat{y}\}$ basis.
\begin{description}
\item[Scenario 1] In this scenario, there is a homothet of $P$ with $a$ on its top left corner and $b$ on its lowest short side. Formally, the slope between $a$ and $b$ satisfies
$$S\in  \left(-\infty, \frac{\cos(\theta_0)-A}{\sin(\theta_0)}\right].$$
In this scenario, the usual basis $\{(1,0),(0,1)\}$ will be sent to $\hat{x}:= (0,-1)$, $\hat{y}:=\left(\sin(\theta_0),\cos(\theta_0)\right)$ (observe that $\frac{1}{A} \hat{x}_b \geq \hat{y}_b$). 
The interior counterclockwise angle $\theta$ between $\hat{x}$ and $\hat{y}$ is then $\theta = \pi- \theta_0$.

In Figure~\ref{fig:para2},
the slope between $a$ and $b_1$ falls into this case. As seen on the first column of Figure~\ref{fig:para2}, there is a homothet of $P$ with $a$ on its top left corner and $b_1$ on its lowest short side. 

As an example, suppose the vertices of $P$ have coordinates $(0,0)$, $(\sin(\theta_0),\cos(\theta_0))$, $(0,-2)$ and $(\sin(\theta_0),\cos(\theta_0)-2)$, respectively. Then, in the new basis, the vertices of $\hat{P}$ have coordinates $(0,0)$, $(0,1)$, $(2,0)$ and $(2,1)$, respectively.

\item[Scenario 2] In this scenario, there is a homothet of $P$ with $a$ on its top left corner and $b$ on its furthest long side. Formally, the slope between $a$ and $b$ satisfies $$S\in \left(\frac{\cos(\theta_0)-A}{\sin(\theta_0)}, \frac{\cos(\theta_0)}{\sin(\theta_0)}\right].$$
In this scenario, the usual basis $\{(1,0),(0,1)\}$ will be sent to $\hat{x}:=\left(\sin(\theta_0),\cos(\theta_0)\right)$, $\hat{y}:=(0,-1)$ (observe that $A \hat{x}_b \geq \hat{y}_b$).
The interior counterclockwise angle $\theta$ between $\hat{x}$ and $\hat{y}$ is then $\theta = \pi- \theta_0$.

In Figure~\ref{fig:para2},
the slope between $a$ and $b_2$ falls into this case. As seen on the first column of Figure~\ref{fig:para2}, there is a homothet of $P$ with $a$ on its top left corner and $b_2$ on its furthest long side. 

As an example, suppose the vertices of $P$ have coordinates $(0,0)$, $(\sin(\theta_0),\cos(\theta_0))$, $(0,-2)$ and $(\sin(\theta_0),\cos(\theta_0)-2)$, respectively. Then, in the new basis, the vertices of $\hat{P}$ have coordinates $(0,0)$, $(1,0)$, $(0,2)$ and $(1,2)$, respectively.

\item[Scenario 3] In this scenario, there is a homothet of $P$ with $a$ on its bottom left corner and $b$ on its furthest long side. Formally, the slope between $a$ and $b$ satisfies $$S\in \left(\frac{\cos(\theta_0)}{\sin(\theta_0)}, \frac{\cos(\theta_0)+A}{\sin(\theta_0)}\right].$$
In this scenario, the usual basis $\{(1,0),(0,1)\}$ will be sent to $\hat{x}:= \left(\sin(\theta_0),\cos(\theta_0)\right)$, $\hat{y}:=(0,1)$ (observe that $A \hat{x}_b \geq \hat{y}_b$). The interior counterclockwise angle $\theta$ between $\hat{x}$ and $\hat{y}$ is then $\theta= \theta_0$.

In Figure~\ref{fig:para2},
the slope between $a$ and $b_3$ falls into this case. As seen on the first column of Figure~\ref{fig:para2}, there is a homothet of $P$ with $a$ on its bottom left corner and $b_3$ on its furthest long side. 

As an example, suppose the vertices of $P$ have coordinates $(0,0)$, $(\sin(\theta_0),\cos(\theta_0))$, $(0,2)$ and $(\sin(\theta_0),\cos(\theta_0)+2)$, respectively. Then, in the new basis, the vertices of $\hat{P}$ have coordinates $(0,0)$, $(1,0)$, $(0,2)$ and $(1,2)$, respectively.

\item[Scenario 4] In this scenario, there is a homothet of $P$ with $a$ on its bottom left corner and $b$ on its highest short side. Formally, the slope between $a$ and $b$ satisfies $$S\in \left(\frac{\cos(\theta_0)+A}{\sin(\theta_0)}, \infty\right).$$
In this scenario, the usual basis $\{(1,0),(0,1)\}$ will be sent to $\hat{x}:= (0,1)$, $\hat{y}:=\big(\sin(\theta_0),\cos(\theta_0)\big)$ (observe that $\frac{1}{A} \hat{x}_b \geq \hat{y}_b$). The interior counterclockwise angle $\theta$ between $\hat{x}$ and $\hat{y}$ is then $\theta= \theta_0$.

In Figure~\ref{fig:para2},
the slope between $a$ and $b_4$ falls into this case. As seen on the first column of Figure~\ref{fig:para2}, there is a homothet of $P$ with $a$ on its  bottom left  corner and $b_4$ on its highest short side. 

As an example, suppose the vertices of $P$ have coordinates $(0,0)$, $(\sin(\theta_0),\cos(\theta_0))$, $(0,2)$ and $(\sin(\theta_0),\cos(\theta_0)+2)$, respectively. Then, in the new basis, the vertices of $\hat{P}$ have coordinates $(0,0)$, $(0,1)$, $(2,0)$ and $(2,1)$, respectively.

\end{description}

In the second column of Figure~\ref{fig:para2}, we highlight the homothets of $\hat{P}$ (with $a$ and $b$ on the boundary) on which our analysis will be based. In the third column we show the homothets of $\hat{P}$ in the new basis $\{ \hat{x},\hat{y} \}$ for each scenario.  Notice that $\hat{P}$ is an axis aligned rectangle in the $\{\hat{x},\hat{y}\}$ basis. We define the W and E sides of $\hat{P}$ to be the two sides parallel  to $\hat{y}$ with the W side having smaller $\hat{x}$-coordinate. Similarly, we define the N and S sides of $\hat{P}$ to be the two sides parallel  to $\hat{x}$ with the N side having larger $\hat{y}$-coordinate. A point on the east edge of $\hat{P}$ is said to be \emph{eastern}. The \emph{directions} N, S, E, W are defined accordingly, refer to the last column of Figure~\ref{fig:para2}.

Observe that the distance between the origin and any point $\alpha\hat{x}+\beta\hat{y}$ in all scenarios is equal to
\begin{align*}
    \|\alpha \hat{x} +\beta \hat{y}  \|_2 = \sqrt{\alpha^2+\beta^2-2\alpha\beta\cos(\pi-\theta)} \leq \sqrt{\alpha^2+\beta^2+2\alpha\beta|\cos(\theta)|} .
\end{align*}
Moreover, observe that the $L_1$-norm in the usual basis is equal to the $L_1$-norm in the $\{\hat{x},\hat{y}\}$-basis.

 As a first step, we consider the triangles $T_1,T_2,...,T_k$ intersecting the line segment $ab$ ordered from $a$ to $b$. We will show, in essence, that the shortest path between $a$ and $b$ in the graph induced by this collection of triangles is a spanning path whose spanning ratio is at most $h(A,\theta_0)$. Note that the triangulation of $\mathcal{P}$ may not contain its convex hull. As such, the sequence of triangles may not exist. For example, the sequence of triangles from $l_1$ to $l_5$ is undefined in Figure \ref{fig:exmaple}. In order to guarantee that this sequence of triangles is well defined, we augment the point set in the following way.
Consider the triangulation of a larger point set $\mathcal{P'}=\mathcal{P}\cup\{p_1,p_2,p_3,p_4\}$, where $\{p_1,p_2,p_3,p_4\}$ is defined as follows.
The points in $\{p_1,p_2,p_3,p_4\}$ are the vertices of a homothet of $P$ containing $\mathcal{P}$ such that the distance from any point in $\mathcal{P}$ to any point in $\{p_1,p_2,p_3,p_4\}$ is arbitrarily large.
Adding these points guarantees that the larger triangulation contains all the edges of the convex hull of $\mathcal{P'}$. Moreover, all the edges of the triangulation of $\mathcal{P'}$ that are not in the triangulation of $\mathcal{P}$, by construction, are arbitrarily long. Therefore, these edges will not be used in any bounded path. In what follows, we bound the length of the shortest path between points in $\mathcal{P}$.

We will denote the triangle with vertices $u$, $v$ and $w$ by $\triangle uvw$. Note that with the additional four points, every face that intersects the line segment $ab$ is bounded and therefore is a triangle. 
 The boundary of each triangle intersects $ab$ twice.
Let $h_i$ and $l_i$ denote the endpoints of the edge in $T_i$ intersecting line segment $ab$ closest to $b$.
Note that $h_i$ is above $ab$ and $l_i$ is below $ab$
(with respect to the $\{\hat{x},\hat{y}\}$ basis\footnote{We do not use the notation $\hat{h}_i$ and 
 $\hat{l}_i$ even though these two points are defined with respect to the $\{\hat{x},\hat{y}\}$ basis. This would make expressions like $\hat{x}_{\hat{h}_i}$ too heavy.}).
Every pair of consecutive triangles along the line segment $ab$ shares two vertices. As such, for all $1 \leq i < k$, either $l_{i}=l_{i+1}$ or $h_{i}=h_{i+1}$. We define $h_{0}=l_0=a$, and we let $l_k = h_k=b$.
Figure~\ref{fig:exmaple} shows the triangles as well as their respective parallelogram $\hat{P}_i$ with $h_i$ and $l_i$ in Scenario 2.
\begin{figure}
    %\centering
   % \captionsetup{justification=centering}
   \centerline{
    \includegraphics[scale=1]{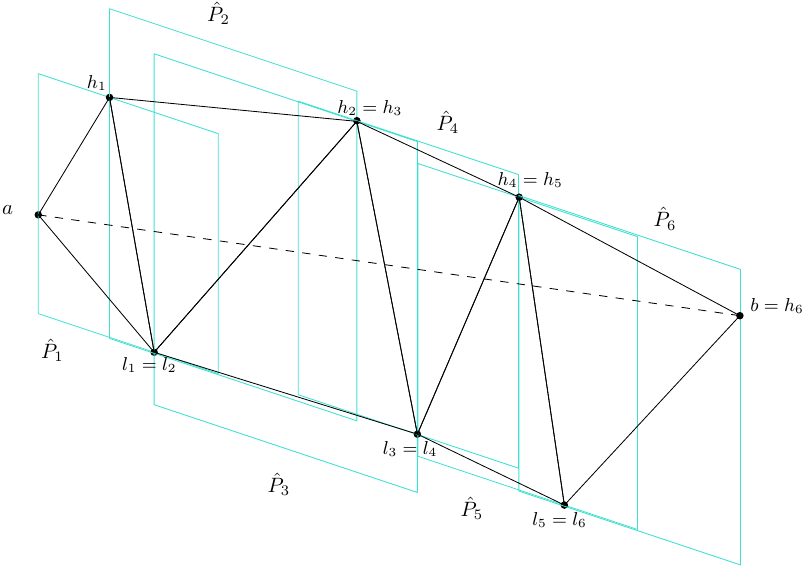}}
    \caption{Illustration of the triangles intersecting $ab$.\label{fig:exmaple}}
\end{figure}
 For each triangle $T_i$ (or $\hat{T}_i$ in the $\{\hat{x},\hat{y}\}$ basis), there is a corresponding scaled translate $P_i$ (or $\hat{P}_i$ in the $\{\hat{x},\hat{y}\}$ basis) of $P$ (or $\hat{P}$ in the $\{\hat{x},\hat{y}\}$ basis). The parallelogram $P_i$ contains no point of $\mathcal{P}$ in its interior and has the three vertices of $T_i$ on its boundary. 
We define the
parameter $L$ as the positive slope of the diagonal of $\hat{P}$ expressed in the $\{\hat{x}, \hat{y}\}$
basis.
The third column of Figure~\ref{fig:para2} illustrates these diagonals\footnote{Note that the slope in the $\{\hat{x},\hat{y}\}$ basis corresponds to the ratio of the sides of $\hat{P}$.}.
If the long side of $\hat{P}$ is vertical, then $L = A$, as in Scenarios 2 and 3. Otherwise, $L = 1/A$, as in Scenarios 1 and 4.

For any point $u$, we denote $\hat{x}_{u}$ and $\hat{y}_u$ to be the coordinates of $u$ in its $\{\hat{x}, \hat{y}\}$ coordinate system. In other words, $u=\hat{x}_u\hat{x}+\hat{y}_u\hat{y}$.
\begin{definition}
     A {\em gentle} edge $(u,v)$ has the property that $|\hat{y}_v-\hat{y}_u| \leq L|{\hat{x}_v-\hat{x}_u}|$. Otherwise, the edge is {\em steep}.
\end{definition}
We next define the notion of an \emph{inductive} parallelogram. 
\begin{definition}
     If the edge $(h_i, l_i)$ is gentle, then $\hat{P}_i$ is defined to be \emph{inductive}. The point with the larger $\hat{x}$-coordinate among $h_i$ and $l_i$ is said to be the {\em inductive point} of  $\hat{P}_i$, denoted by $c$. 
\end{definition}
We use a potential to bound the spanning ratio. This potential is tied to a parallelogram. We highlight the exact relationship between the potential and the parallelogram below. Note that our potential is very similar to the potential introduced by Bonichon et al.~\cite{DBLP:journals/comgeo/BonichonGHP15} and also used by van Renssen et al.~\cite{DBLP:conf/esa/RenssenSSW23} who generalized it to account for the aspect ratio. We further modify this approach to account for the angle $\theta_0$ in a parallelogram. 
Let $d_{\hat{P}_i}(h_i,l_i)$ be the length of the path when moving clockwise from $h_i$ to $l_i$ along the sides of $\hat{P}_i$.
Note that in the usual $x$, $y$-coordinate system,  this path may be counter-clockwise. For example in Figure~\ref{fig:clock}, we show in blue the length of $d_{\hat{P}_{uw}}(u,w)$ in Scenario 2. On the left, we have $\hat{P}_{uw}$ and on the right we have $P_{uw}$ (in the usual $\{x,y\}$ basis). When in the $\{\hat{x},\hat{y}\}$ basis, the length of $d_{\hat{P}_{uw}}(u,w)$ is the length of the path when going from $u$ to $w$, whereas in the usual basis, this length is given by the counter clockwise path from $u$ to $w$ (or the clockwise path from $w$ to $u$).
\begin{figure}
    %\centering
   \captionsetup{justification=centering}
   \centerline{
    \includegraphics[scale=1]{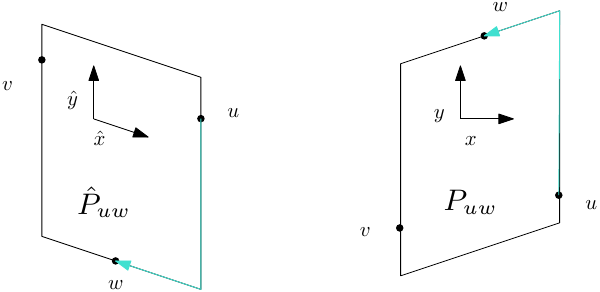}}
    \caption{Illustration of the length of $d_{\hat{P}_{uw}}(u,w)$ in Scenario 2.}
    
    \label{fig:clock}
\end{figure}

\begin{definition}
   Parallelogram $\hat{P}_i$ \emph{has a potential} if $$d_2^T(a,h_i)+d_2^T(a,l_i)+d_{\hat{P}_i}(h_i,l_i)\leq (2+2L)\hat{x}_{i},$$ where $\hat{x}_i$ is the $\hat{x}$-coordinate of the E side of $\hat{P}_i$.
 \end{definition}

Next, we define $P(a,b)$ as the parallelogram $\{\alpha\hat{x}+\beta\hat{y}: \hat{x}_a \leq \alpha \leq \hat{x}_b,\hat{y}_a \leq \beta \leq \hat{y}_b \}$
(refer to Figure~\ref{fig:P(a,b)}).
\begin{figure}
    %\centering
   % \captionsetup{justification=centering}
   \centerline{
    \includegraphics[scale=1]{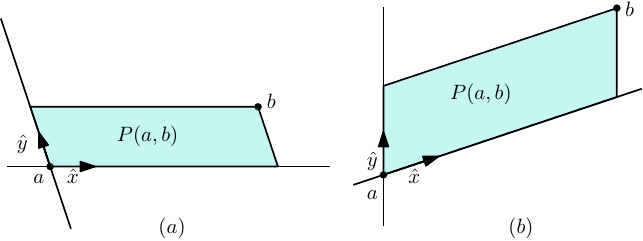}}
    \caption{Illustration of 
 $P(a,b)$ in Scenarios 1 and 3.\label{fig:P(a,b)}}
\end{figure}
Observe that in general, $P(a,b)$ is not a homothet of  $P$ or of $\hat{P}$.

The following lemma describes how the potential of one parallelogram is related to the next.
\begin{lemma}\label{lem:lem4}
    If $(a,b)$ is not an edge in $T$ and parallelogram $P(a,b)$ contains no point of $\mathcal{P}$ other than $a$ and $b$, then $\hat{P}_1$  has a potential. Furthermore, if, for any $1 \leq i <k$, $\hat{P}_i$ has a potential but is not inductive, then $\hat{P}_{i+1}$ has a potential.
\end{lemma}
\begin{proof}
   Firstly, $a$ cannot lie on the S side of $\hat{P}_1$ as it would imply that $l_1$ is on the E side and hence in the interior of $P(a,b)$, which we assumed to be empty. Refer to Figure~\ref{fig:P1}. Therefore, $a$ must be on the W side of $\hat{P}_1$, while $h_1$ must lie on the N or the E side, and $l_1$ must lie on the S or E side.
    Observe that the perimeter of $\hat{P}_1$, which is $(2+2L)\hat{x}_{1}$ in the $\{\hat{x}, \hat{y}\}$ basis, upper bounds $d_2^T(a,h_1)+d_2^T(a,l_1)+d_{\hat{P}_1}(h_1,l_1)$. 
    \begin{figure}
    %\centering
   \captionsetup{justification=centering}
   \centerline{
    \includegraphics[scale=1]{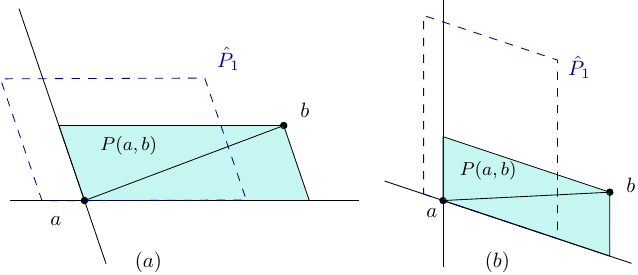}}
    \caption{Example showing the positioning of $\hat{P}_1$ if $a$ is on the S side in Scenario 1 on the left $(a)$, and Scenario 2 on the right $(b)$.}
    \label{fig:P1}
\end{figure}

     From now on, suppose that for some $1\leq i <k$, $\hat{P}_i$ has a potential but is not inductive. Since $\hat{P}_i$ is not inductive we know that $(h_i, l_i)$ is steep. It suffices to consider the case where $\hat{x}_{l_i}< \hat{x}_{h_i}$ since the case where $\hat{x}_{h_i}< \hat{x}_{l_i}$ can be handled similarly. 
\begin{figure}
\begin{center}
%\centering
    \captionsetup{justification=centering}
    \includegraphics[scale=1]{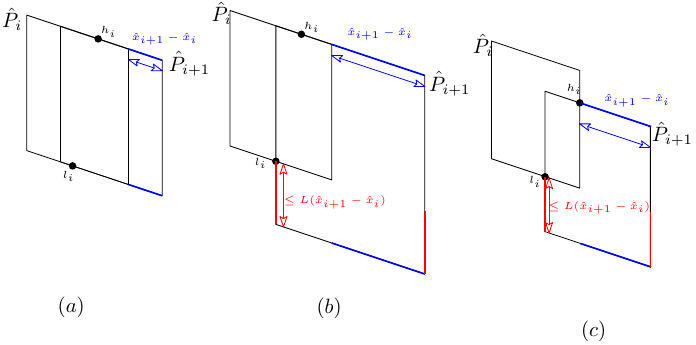}
    \caption{Illustration of the proof of Lemma~\ref{lem:lem4} for Scenario 2.\label{fig:lem4}}
\end{center}
\end{figure}

     Since $\hat{x}_{l_i}< \hat{x}_{h_i}$, we know that $l_i$ must be on the S side of $\hat{P}_{i}$ while $h_i$ can be on the N or E side of $\hat{P}_{i}$ . If $h_i$ is on the N side of $\hat{P}_{i}$, then since $\hat{x}_{l_i}< \hat{x}_{h_i}$, then $h_i$ must be on the N side of $\hat{P}_{i+1}$  and $l_i$ is either on the W or S side of $\hat{P}_{i+1}$. If $l_i$ is on the S side of $\hat{P}_{i+1}$ (refer to Figure~\ref{fig:lem4} Case (a)), we have that $\hat{P}_{i+1}$ is just a translate of $\hat{P}_{i}$, as such we have that 
    \begin{align}~\label{eq:lem4.1}
        d_{\hat{P}_{i+1}}(h_i,l_i)-d_{\hat{P}_{i}}(h_i,l_i)=2(\hat{x}_{i+1}-\hat{x}_i).
    \end{align}
     Next, in the case where $l_i$ is on the W side of $\hat{P}_{i+1}$, (refer to Figure~\ref{fig:lem4} Case (b)), then we have that 
    \begin{align}~\label{eq:lem4.2}
        d_{\hat{P}_{i+1}}(h_i,l_i)-d_{\hat{P}_{i}}(h_i,l_i)\leq (2+2L)(\hat{x}_{i+1}-\hat{x}_i).
    \end{align}
    Consider now the case where $h_i$ is on the E side of $\hat{P}_{i}$. Recall that $\hat{x}_{l_i}< \hat{x}_{h_i}$, hence  $h_i$ must lie the N side of $\hat{P}_{i+1}$ and $l_i$ is either on the S side or on the W side of $\hat{P}_{i+1}$ (refer to Figure~\ref{fig:lem4} Case (c)). If $l_i$ is on the S side of $\hat{P}_{i+1}$ then~\eqref{eq:lem4.1} holds. If $l_i$ is on the W side of $\hat{P}_{i+1}$ then~\eqref{eq:lem4.2} holds. In all cases we have that $$d_{\hat{P}_{i+1}}(h_i,l_i)-d_{\hat{P}_{i}}(h_i,l_i)\leq (2+2L)(\hat{x}_{i+1}-\hat{x}_i).$$
 Since $\hat{P}_{i}$ has potential we know that
    \begin{align*}
        d_2^T(a,h_i)+d_2^T(a,l_i)+d_{\hat{P}_{i}}(h_i,l_i) \leq (2+2L)\hat{x}_{i} .
    \end{align*}
    Observe that 
    \begin{align*}
        &\ d_2^T(a,h_i)+d_2^T(a,l_i)+d_{\hat{P}_{i}}(h_i,l_i)+d_{\hat{P}_{i+1}}(h_i,l_i)-d_{\hat{P}_{i}}(h_i,l_i)\\
        =&\ d_2^T(a,h_i)+d_2^T(a,l_i)+d_{\hat{P}_{i+1}}(h_i,l_i).
    \end{align*}
    Putting the above equalities and inequalities together we get
    \begin{align}~\label{eq:potential}
        \begin{split}
            d_2^T(a,h_i)+d_2^T(a,l_i)+d_{\hat{P}_{i+1}}(h_i,l_i) &\leq (2+2L)\hat{x}_{i} +(2+2L)(\hat{x}_{i+1}-\hat{x}_i)\\
        &= (2+2L)\hat{x}_{i+1}
        \end{split}
    \end{align}
    Assume $T_{i+1}=\triangle h_i h_{i+1} l_i$ (where $l_i=l_{i+1}$); in other words, $(h_i, h_{i+1})$ is an edge of $T$ with $h_{i+1}$ lying somewhere on the boundary of $\hat{P}_{i+1}$ between $h_i$ and $l_i$, when moving clockwise from $h_i$ to $l_i$. By the triangle inequality, we have  $d_2(h_i,h_{i+1})\leq d_{\hat{P}_{i+1}}(h_i, h_{i+1})$ and
    \begin{align}\label{eq:lem4.3}
        d_2^T(a,h_{i+1})\leq d_2^T(a,h_{i})+d_2(h_i,h_{i+1}) .
    \end{align}
    We also know that
    \begin{align}\label{eq:lem4.4}
        d_{\hat{P}_{i+1}}(h_{i},h_{i+1})\leq d_{\hat{P}_{i+1}}(h_{i},l_{i}).
    \end{align}
    This gives us
    \begin{align}\label{eq:lem4.5}
    \begin{split}
        d_2(h_i,h_{i+1})+d_{\hat{P}_{i+1}}(h_{i+1},l_i) &\leq d_{\hat{P}_{i+1}}(h_{i},h_{i+i})+d_{\hat{P}_{i+1}}(h_{i+1},l_{i})\\
        &=d_{\hat{P}_{i+1}}(h_{i},l_{i}) .
    \end{split}  
    \end{align}
    
    Using inequalities~\eqref{eq:potential},~\eqref{eq:lem4.3},~\eqref{eq:lem4.4} and~\eqref{eq:lem4.5}, we get the following:
    \begin{align*}
        d_2^T(a,h_{i+1})+d_2^T(a, l_{i+1})+d_{\hat{P}_{i+1}}(h_{i+1},l_{i+1}) \leq &\ d_2^T(a,h_i)+d_2(h_i,h_{i+1})\\
        &\ +d_2^T(a,l_i)
        +d_{\hat{P}_{i+1}}(h_{i+1},l_i)\\
        \leq &\ d_2^T(a,h_i)+d_2^T(a,l_i)\\
        &\ +d_{\hat{P}_{i+1}}(h_{i},l_i)\\
        \leq &\ (2+2L)\hat{x}_{i+1},
    \end{align*}
     as required to show $\hat{P}_{i+1}$ has a potential. Note that here $l_{i+1}=l_i$. The argument for the case when $T_{i+1}=\triangle h_{i+1} l_i l_{i+1}$, where $h_i=h_{i+1}$ is symmetric.
\end{proof}

     We now use this notion of potential to prove an upper bound on the $\hat{x}$-distance between $a$ and the inductive point of a parallelogram with potential.

     \begin{lemma}~\label{lem:lem5}
         If parallelogram $\hat{P}_{i}$ has a potential and its inductive point $c$ (either $c=h_i$ or $c=l_i$) lies on the E side of $\hat{P}_{i}$, then $$d_2^T(a,c)\leq (1+L)\hat{x}_c .$$
     \end{lemma}

     \begin{proof}
     Without loss of generality, assume $c=h_i$.
     Observe that, since $h_i$ is eastern, we have
     \begin{align}
     \label{eq:lem5}
     \hat{x}_c=\hat{x}_{h_i}=\hat{x}_i .
     \end{align}
     
     Moreover, since $\hat{P}_{i}$ has a potential, we have (1)
     $d_2^T(a,h_i)\leq (1+L)\hat{x}_{i}=(1+L)\hat{x}_{h_i}$ or (2) $d_2^T(a,l_i)+d_{\hat{P}_{i}}(h_i,l_i) \leq (1+L)\hat{x}_{i}$.

     In the first case, we find
     $$d_2^T(a,h_i)\leq (1+L)\hat{x}_{i}=(1+L)\hat{x}_c$$
     by~\eqref{eq:lem5}.
     In the second case, since $(l_i,h_i)$ is an edge in $T$, by the triangle inequality we get
\begin{align*}
    d_2^T(a,h_i)\leq d_2^T(a,l_i)+d_2(l_i,h_i)\leq d_2^T(a,l_i)+d_{\hat{P}_{i}}(h_i,l_i)&\leq (1+L)\hat{x}_{i}\\
    &=(1+L)\hat{x}_c
\end{align*}
by~\eqref{eq:lem5}.
\end{proof}

\begin{definition}~\label{def:def6}
 Let $1\leq j \leq k$. The \emph{maximal high path ending at $h_j$} and the \emph{maximal low path ending at $l_j$} are defined as follows:\\
 If $h_j$ is eastern in $\hat{P}_{j}$, the maximal high path ending at $h_j$ is simply $h_j$; otherwise, it is the path $h_i, h_{i+1},...,h_j$ such that $h_{i+1},...,h_j$ are not eastern in respectively, $\hat{P}_{i+1},..,\hat{P}_{j}$ and either $i=0$ or $h_i$ is eastern in $\hat{P}_{i}$.\\
 If $l_j$ is eastern in $\hat{P}_{j}$, the maximal low path ending at $l_j$ is simply $l_j$; otherwise, it is the path $l_i, l_{i+1},...,l_j$ such that $l_{i+1},...,l_j$ are not eastern in respectively, $\hat{P}_{i+1},..,\hat{P}_{j}$ and either $i=0$ or $l_i$ is eastern in $\hat{P}_{i}$.
\end{definition}

The following lemma gives an upper bound on the length of maximal high and low paths.

\begin{lemma}~\label{lem:lem7}
    Suppose $h_i,h_{i+1},...,h_j$ is a maximal high path. Then we have $$d_2^T(h_i,h_j)\leq (\hat{x}_{h_j}-\hat{x}_{h_i})+(\hat{y}_{h_j}-\hat{y}_{h_i}).$$ On the other hand, the following inequality holds for a maximal low path $l_i,l_{i+1},...,l_j$:  $$d_2^T(l_i,l_j)\leq (\hat{x}_{l_j}-\hat{x}_{l_i})+(\hat{y}_{l_i}-\hat{y}_{l_j}).$$ 
\end{lemma}

\begin{proof}
     Although $h_i$ may be E, we note that the remaining vertices of the maximal high path $h_{i+1},..,h_j$ cannot be E. As such we have a succession of WN edges (by the general position assumption). As such, we have that $\hat{y}_{h_i} < \hat{y}_{h_{i+1}} < ... < \hat{y}_{h_j}$. By the triangle inequality, we have the following for all $i\leq k\leq j$: $$d_2^T(h_k,h_{k+1})\leq (\hat{x}_{h_{k+1}}-\hat{x}_{h_{k}})+(\hat{y}_{h_{k+1}}-\hat{y}_{h_k}) .$$ We obtain
    \begin{align*}
        d_2^T(h_i,h_j)\leq \sum_{k=i}^j d_2^T(h_k,h_{k+1}) \leq \sum_{k=i}^j  (\hat{x}_{h_{k+1}}-\hat{x}_{h_{k}})+(\hat{y}_{h_{k+1}}-\hat{y}_{h_k}) \leq (\hat{x}_{h_j}-\hat{x}_{h_i})+(\hat{y}_{h_j}-\hat{y}_{h_i}).
    \end{align*}
    The bound on the length of the maximal low path can be shown using a symmetrical argument
\end{proof}

Next we will prove the Crossing Lemma, which yields bounds on the distance in the triangulation from $a$ to the first inductive point.
\begin{lemma}\label{lem:crossing}
    (Crossing Lemma) Assume $P(a,b)$ does not contain any other vertices of $\mathcal{P}$ and $(a,b)$ is not an edge in the parallelogram Delaunay graph.  Then in each case, we have the following: 
    \begin{enumerate}
        \item[(1)]~\label{eq:crlem1}If no parallelogram in $\hat{P}_{1},\hat{P}_{2},...,\hat{P}_{k}$ is inductive then $$d_2^T(a,b)\leq \left(L+\sqrt{1+L^2+2L|\cos(\theta)|}\right)\hat{x}_b+\hat{y}_b$$
        \item[(2)] Otherwise, let $\hat{P}_{j}$ be the first inductive parallelogram in the sequence $\hat{P}_{1},\hat{P}_{2},...,\hat{P}_{k-1}$.
        \begin{enumerate}
            \item[(a)]~\label{eq:crlem2.a} If $L=A$ (as in Scenarios 2 and 3) and $h_j$ is the inductive point of $\hat{P}_{j}$ then, $$d_2^T(a,h_j)+(\hat{y}_{h_j}-\hat{y}_b) \leq \left(A+\sqrt{1+A^2+2A|\cos(\theta)|}\right)\hat{x}_{h_j}.$$
            \item[(b)] If $L=A$ (as in Scenarios 2 and 3) and $l_j$ is the inductive point of $\hat{P}_{j}$ then, $$ d_2^T(a,l_j)-\hat{y}_{l_j} \leq \left(A+\sqrt{1+A^2+2A|\cos(\theta)|}\right)\hat{x}_{l_j}.$$
            \item[(c)] If $L=1/A$ (as in Scenarios 1 and 4) and $h_j$ is the inductive point of $\hat{P}_{j}$ then, $$d_2^T(a,h_j)+A(\hat{y}_{h_j}-\hat{y}_b) \leq \left(1+\sqrt{1+\frac{1}{A^2}+\frac{2|\cos(\theta)|}{A}}\right)\hat{x}_{h_j}.$$
            \item[(d)] If $L=1/A$ (as in Scenarios 1 and 4) and $l_j$ is the inductive point of $\hat{P}_{j}$ then, $$ d_2^T(a,l_j)-A\hat{y}_{l_j} \leq \left(1+\sqrt{1+\frac{1}{A^2}+\frac{2|\cos(\theta)|}{A}}\right)\hat{x}_{l_j}.$$
        \end{enumerate}
    \end{enumerate}
\end{lemma}

\begin{proof}
\begin{enumerate}
\item[(1)] 
  Suppose $P(a,b)$ is empty and there is no inductive parallelogram in the sequence $\hat{P}_{1},\hat{P}_{2},...,\hat{P}_{k}$. Then by Lemma~\ref{lem:lem4}, $\hat{P}_{k}$ must have a potential since $\hat{P}_{1}$ has a potential. Next, $b$ must lie on the E side of $\hat{P}_{k}$ by the general position assumption. Lemma~\ref{lem:lem5} gives the following bound:
    \begin{align*}
        d_2^T(a,b)\leq (1+L)\hat{x}_k=(1+L)\hat{x}_b \leq \left(L+\sqrt{1+L^2+2L|\cos(\theta)|}\right)\hat{x}_b+\hat{y}_b .
    \end{align*}
\item[(2)]
\begin{enumerate}
\item[(a)] In this case, we are either in Scenario 2 or 3.

  Suppose $\hat{P}_{j}$ is the first inductive parallelogram in $\hat{P}_{1},\hat{P}_{2},...,\hat{P}_{k-1}$  and we assume that the inductive point of $\hat{P}_{j}$ is $c=h_j$.
     
    By Lemma ~\ref{lem:lem4}, every parallelogram $\hat{P}_{i}$ with $i\leq j$ has a potential. Note that in this case, $l_j$ is to the left of $h_j$.
    Since the edge $(l_j,h_j)$ is gentle it follows that the vertical distance between $l_j$ and $h_j$
    can be bounded:
    \begin{align}~\label{eq:lem8_2a1}
    (\hat{y}_{h_j}-\hat{y}_{l_j})  \leq A(\hat{x}_{h_j}-\hat{x}_{l_j}) \iff (\hat{y}_{h_j}-\hat{y}_{l_j})+A \hat{x}_{l_j}\leq A \hat{x}_{h_j}.
\end{align}

Moreover, we can bound the length of the edge $(l_j,h_j)$ by
\begin{align}~\label{eq:lem8_2a2}
    \begin{split}
        d_2(l_j,h_j)&\leq \sqrt{1+A^2-2A\cos(\pi-\theta)}(\hat{x}_{h_j}-\hat{x}_{l_j})\\
        &\leq \sqrt{1+A^2+2A|\cos(\theta)|}(\hat{x}_{h_j}-\hat{x}_{l_j}).
    \end{split}
\end{align}

Let $l_i,l_{i+1},...,l_{j-1}=l_j$ be the maximal low path ending at $l_j$. Note that by Lemma~\ref{lem:lem7} we have
\begin{align}~\label{eq:lem8_2a3}
    d_2^T(l_i,l_j)&\leq (\hat{x}_{l_j}-\hat{x}_{l_i})+(\hat{y}_{l_i}-\hat{y}_{l_j}) .
\end{align}

Next, note that either $l_i=l_0=a$ or $l_i$ is an eastern point in $\hat{P}_{i}$ that has a potential and Lemma~\ref{lem:lem5} applies. In this case we have that since $l_i$ is eastern in $\hat{P}_{i}$ we have:
\begin{align}~\label{eq:lem8_2a4}
    \begin{split}
       d_2^T(a,l_i)
       &\leq (1+A)\hat{x}_{i}\\
       &= (1+A)\hat{x}_{l_i} 
    \end{split}
\end{align}

Using inequalities ~\eqref{eq:lem8_2a1},~\eqref{eq:lem8_2a2},~\eqref{eq:lem8_2a3} and~\eqref{eq:lem8_2a4}, together with the triangle inequality, we get
\begin{align*}
    &\ d_2^T(a,h_j)+(\hat{y}_{h_j}-\hat{y}_b)\\
    \leq&\ d_2^T(a,l_i)+d_2^T(l_i,l_j)+d_2(l_j,h_j)+(\hat{y}_{h_j}-\hat{y}_b)\\
    \leq&\ (1+A)\hat{x}_{l_i} +(\hat{x}_{l_j}-\hat{x}_{l_i})+(\hat{y}_{l_i}-\hat{y}_{l_j})   + \sqrt{1+A^2+2A|\cos(\theta)|}(\hat{x}_{h_j}-\hat{x}_{l_j})+ \hat{y}_{h_j}-\hat{y}_b\\
    =&\ A\hat{x}_{l_i}+\sqrt{1+A^2+2A|\cos(\theta)|}\hat{x}_{h_j}+(\hat{y}_{l_i}-\hat{y}_{l_j}) -\left(1+\sqrt{1+A^2+2A|\cos(\theta)|}\right)\hat{x}_{l_j}  +\hat{y}_{h_j}-\hat{y}_b\\
    \leq&\ A\hat{x}_{l_i}+ \sqrt{1+A^2+2A|\cos(\theta)|}\hat{x}_{h_j}+ \hat{y}_{h_j}-\hat{y}_{l_j}\\
     &\ \phantom{.}\hskip 4cm \text{since $\hat{y}_{l_i}-\hat{y}_b<0$ and $-\hat{x}_{l_j}<0$}\\
    \leq&\ \sqrt{1+A^2+2A|\cos(\theta)|}\hat{x}_{h_j} + A\hat{x}_{l_j}+\hat{y}_{h_j}-\hat{y}_{l_j}\\
    \leq&\ \sqrt{1+A^2+2A|\cos(\theta)|}\hat{x}_{h_j} + A\hat{x}_{h_j}\\
     &\ \phantom{.}\hskip 4cm \text{since edge $(l_j,h_j)$ is gentle}\\
    =&\ \left(A+\sqrt{1+A^2+2A|\cos(\theta)|}\right)\hat{x}_{h_j} .
\end{align*}

\item[(b)] In this case,  we are either in Scenario 2 or 3.

   Suppose $\hat{P}_{j}$ is the first inductive parallelogram in $\hat{P}_{1},\hat{P}_{2},...,\hat{P}_{k-1}$  and we assume that the inductive point of $\hat{P}_{j}$ is $c=l_j$.
  
    By Lemma ~\ref{lem:lem4}, every parallelogram $\hat{P}_{i}$, for $i\leq j$ has a potential. Note that in this case, $h_j$ is to the left of $l_j$. Since the edge $(l_j,h_j)$ is gentle it follows that the vertical distance between $l_j$ and $h_j$ can be bounded:
\begin{align}~\label{eq:lem8_2b1} 
    (\hat{y}_{h_j}-\hat{y}_{l_j})  \leq A(\hat{x}_{l_j}-\hat{x}_{h_j}) \iff (\hat{y}_{h_j}-\hat{y}_{l_j})+A \hat{x}_{h_j}\leq A \hat{x}_{l_j}.
\end{align}

Moreover, we can bound the length of the edge $(l_j,h_j)$ by
\begin{align}~\label{eq:lem8_2b2} 
\begin{split}
    d_2(l_j,h_j)&\leq \sqrt{1+A^2-2A\cos(\theta)}(\hat{x}_{l_j}-\hat{x}_{h_j})\\
    &\leq  \sqrt{1+A^2+2A|\cos(\theta)|}(\hat{x}_{l_j}-\hat{x}_{h_j})
\end{split}
\end{align}

Let $h_i,h_{i+1},...,h_{j-1}=h_j$ be the maximal high path ending at $h_j$. Note that by Lemma~\ref{lem:lem7} we have
\begin{align}~\label{eq:lem8_2b3} 
    d_2^T(h_i,h_j)\leq (\hat{x}_{h_j}-\hat{x}_{h_i})+(\hat{y}_{h_j}-\hat{y}_{h_i})
\end{align}

Next, note that either $h_i=h_0=a$ or $h_i$ is an eastern point in $\hat{P}_{i}$ that has a potential and Lemma~\ref{lem:lem5} applies. In this case we have that since $h_i$ is eastern in $\hat{P}_{i}$ we have:
\begin{align}~\label{eq:lem8_2b4} 
    \begin{split}
       d_2^T(a,h_i)
       &\leq (1+A)\hat{x}_{i}\\
       &= (1+A)\hat{x}_{h_i} 
    \end{split}
\end{align}

Using inequalities~\eqref{eq:lem8_2b1},~\eqref{eq:lem8_2b2},~\eqref{eq:lem8_2b3} and~\eqref{eq:lem8_2b4}, together with the triangle inequality, we get
\begin{align*}
   &\ d_2^T(a,l_j)-\hat{y}_{l_j}\\
   \leq &\ d_2^T(a,h_i)+d_2^T(h_i,h_j)+d_2(h_j,l_j)-\hat{y}_{l_j}\\
    \leq &\ (1+A)\hat{x}_{h_i} + (\hat{x}_{h_j}-\hat{x}_{h_i})+(\hat{y}_{h_j}-\hat{y}_{h_i})  + \sqrt{1+A^2+2A|\cos(\theta)|}(\hat{x}_{l_j}-\hat{x}_{h_j})-\hat{y}_{l_j}\\
    = &\ A\hat{x}_{h_i}+\hat{x}_{h_j}+\sqrt{1+A^2+2A|\cos(\theta)|}(\hat{x}_{l_j}-\hat{x}_{h_j})+\hat{y}_{h_j}-\hat{y}_{l_j}  -\hat{y}_{h_i}\\
    = &\ \sqrt{1+A^2+2A|\cos(\theta)|}\hat{x}_{l_j}+A\hat{x}_{h_i}+\hat{y}_{h_j} -\hat{y}_{l_j}  -\left(\sqrt{1+A^2+2A|\cos(\theta)|}-1\right)\hat{x}_{h_j}-\hat{y}_{h_i}\\
    \leq &\ \sqrt{1+A^2+2A|\cos(\theta)|}\hat{x}_{l_j} +A\hat{x}_{h_i}+ \hat{y}_{h_j}-\hat{y}_{l_j}\\
     &\ \phantom{.}\hskip 4cm \text{since both $-\hat{y}_{h_i}<0$ and $-\hat{x}_{h_j}<0$}\\
    \leq &\ \sqrt{1+A^2+2A|\cos(\theta)|}\hat{x}_{l_j} +A\hat{x}_{h_j}+ \hat{y}_{h_j}-\hat{y}_{l_j}\\
     &\ \phantom{.}\hskip 4cm \text{since $\hat{x}_{h_i}\leq \hat{x}_{h_j}$}\\
     \leq &\ \sqrt{1+A^2+2A|\cos(\theta)|}\hat{x}_{l_j} +A\hat{x}_{l_j}\\
    &\ \phantom{.}\hskip 4cm \text{since edge  $(h_j,l_j)$ is gentle}\\
    \leq &\ \left(A+\sqrt{1+A^2+2A|\cos(\theta)|}\right)\hat{x}_{l_j}.\\
\end{align*}

\item[(c)] In this case, we are either in Scenario 1 or 4.

     Suppose $\hat{P}_{j}$ is the first inductive parallelogram in $\hat{P}_{1},\hat{P}_{2},...,\hat{P}_{k-1}$  and we assume that the inductive point of $\hat{P}_{j}$ is $c=h_j$.

By Lemma ~\ref{lem:lem4}, every parallelogram $\hat{P}_{i}$, for $i\leq j$ has a potential. Note that in this case, $l_j$ is to the left of $h_j$. Since the edge $(l_j,h_j)$ is gentle it follows that the vertical distance between $l_j$ and $h_j$ can be bounded:
\begin{align}~\label{eq:lem8_2c1}
    (\hat{y}_{h_j}-\hat{y}_{l_j})  \leq \frac{1}{A}(\hat{x}_{h_j}-\hat{x}_{l_j}) \iff A(\hat{y}_{h_j}-\hat{y}_{l_j}) + \hat{x}_{l_j}\leq \hat{x}_{h_j}.
\end{align}

Moreover, we can bound the length of the edge $(l_j,h_j)$ by
\begin{align}~\label{eq:lem8_2c2}
\begin{split}
    d_2(l_j,h_j)&\leq \sqrt{1+\frac{1}{A^2}-\frac{2\cos(\pi-\theta)}{A}}(\hat{x}_{h_j}-\hat{x}_{l_j})\\
    &\leq \sqrt{1+\frac{1}{A^2}+\frac{2|\cos(\theta)|}{A}}(\hat{x}_{h_j}-\hat{x}_{l_j}) .
\end{split} 
\end{align}
Let $l_i,l_{i+1},...,l_{j-1}=l_j$ be the maximal low path ending at $l_j$. Note that by Lemma~\ref{lem:lem7} we have
\begin{align}~\label{eq:lem8_2c3}
    d_2^T(l_i,l_j)&\leq (\hat{x}_{l_j}-\hat{x}_{l_i})+(\hat{y}_{l_i}-\hat{y}_{l_j}) .
\end{align}

Next, note that either $l_i=l_0=a$ or $l_i$ is an eastern point in $\hat{P}_{i}$ that has a potential and Lemma~\ref{lem:lem5} applies. In this case we have that since $l_i$ is eastern in $\hat{P}_{i}$ we have
\begin{align}~\label{eq:lem8_2c4}
    \begin{split}
       d_2^T(a,l_i)
       &\leq (1+1/A)\hat{x}_{i}\\
       &= (1+1/A)\hat{x}_{l_i} .
    \end{split} 
\end{align}

Using inequalities~\eqref{eq:lem8_2c1}, ~\eqref{eq:lem8_2c2},  ~\eqref{eq:lem8_2c3} and  ~\eqref{eq:lem8_2c4}, together with the triangle inequality, we get
\begin{align*}
    &\ d_2^T(a,h_j)+A(\hat{y}_{h_j}-\hat{y}_b)\\
    \leq &\ d_2^T(a,l_i)+d_2^T(l_i,l_j)+d_2(l_j,h_j)+A(\hat{y}_{h_j}-\hat{y}_b)\\
    \leq&\ (1+1/A)\hat{x}_{l_i} +(\hat{x}_{l_j}-\hat{x}_{l_i})+(\hat{y}_{l_i}-\hat{y}_{l_j}) +\sqrt{1+\frac{1}{A^2}+\frac{2|\cos(\theta)|}{A}}(\hat{x}_{h_j}-\hat{x}_{l_j})
      + A(\hat{y}_{h_j}-\hat{y}_b)\\
    = &\ 1/A \hat{x}_{l_i} +\hat{x}_{l_j}+\sqrt{1+\frac{1}{A^2}+\frac{2|\cos(\theta)|}{A}}(\hat{x}_{h_j}-\hat{x}_{l_j})  +(\hat{y}_{l_i}-\hat{y}_{l_j})
     +A(\hat{y}_{h_j}-\hat{y}_b)\\
    \leq&\ 1/A   \hat{x}_{l_i}+\hat{x}_{l_j}+\sqrt{1+\frac{1}{A^2}+\frac{2|\cos(\theta)|}{A}}(\hat{x}_{h_j}-\hat{x}_{l_j})  +A(\hat{y}_{l_i}-\hat{y}_{l_j})+A(\hat{y}_{h_j}-\hat{y}_b)\\
    &\ \phantom{.}\hskip 5cm \text{since   $A>1$}\\
    = &\ 1/A   \hat{x}_{l_i}+\sqrt{1+\frac{1}{A^2}+\frac{2|\cos(\theta)|}{A}}\hat{x}_{h_j}+ A(\hat{y}_{l_i}-\hat{y}_{l_j}) +A(\hat{y}_{h_j}-\hat{y}_b)\\ \displaybreak
    &\
     -(1+\sqrt{1+\frac{1}{A^2}+\frac{2|\cos(\theta)|}{A}})\hat{x}_{l_j}\\ 
    \leq&\ 1/A   \hat{x}_{l_i}+\sqrt{1+\frac{1}{A^2}+\frac{2|\cos(\theta)|}{A}}\hat{x}_{h_j}+A(\hat{y}_{l_i}-\hat{y}_{l_j}) +A(\hat{y}_{h_j}-\hat{y}_b)\\
    &\ \phantom{.}\hskip 5cm \text{since   $-\hat{x}_{l_j}<0$}\\
    = &\ 1/A   \hat{x}_{l_i}+\sqrt{1+\frac{1}{A^2}+\frac{2|\cos(\theta)|}{A}}\hat{x}_{h_j}+A(\hat{y}_{h_j}-\hat{y}_{l_j}) +A(\hat{y}_{l_i}-\hat{y}_b)\\
    \leq &\ \sqrt{1+\frac{1}{A^2}+\frac{2|\cos(\theta)|}{A}}\hat{x}_{h_j}+1/A   \hat{x}_{l_i}+A(\hat{y}_{h_j}-\hat{y}_{l_j})\\ 
    &\ \phantom{.}\hskip 5cm \text{since $(\hat{y}_{l_i}-\hat{y}_b)<0$}\\
    \leq&\ \sqrt{1+\frac{1}{A^2}+\frac{2|\cos(\theta)|}{A}}\hat{x}_{h_j}+1/A   \hat{x}_{l_j}+A(\hat{y}_{h_j}-\hat{y}_{l_j})\\
    &\ \phantom{.}\hskip 5cm \text{since $\hat{x}_{l_i}\leq\hat{x}_{l_j}$}\\
    \leq&\ \sqrt{1+\frac{1}{A^2}+\frac{2|\cos(\theta)|}{A}}\hat{x}_{h_j}+ \hat{x}_{l_j}+A(\hat{y}_{h_j}-\hat{y}_{l_j})\\
    &\ \phantom{.}\hskip 5cm \text{since $A>1$}\\
    \leq&\ \left(\sqrt{1+\frac{1}{A^2}+\frac{2|\cos(\theta)|}{A}}\right)\hat{x}_{h_j}+\hat{x}_{h_j}\\ 
    &\ \phantom{.}\hskip 5cm \text{since edge $(l_j,h_j)$ is gentle}\\
    \leq&\ \left(1+ \sqrt{1+\frac{1}{A^2}+\frac{2|\cos(\theta)|}{A}}\right)\hat{x}_{h_j} .
\end{align*}

\item[(d)] In this case, we are either in Scenario 1 or 4.

   Suppose $\hat{P}_{j}$ is the first inductive parallelogram in $\hat{P}_{1},\hat{P}_{2},...,\hat{P}_{k-1}$  and we assume that the inductive point of $\hat{P}_{j}$ is $c=l_j$.

By Lemma ~\ref{lem:lem4}, every parallelogram $\hat{P}_{i}$, for $i\leq j$ has a potential. Note that in this case, $h_j$ is to the left of $l_j$. Since the edge $(l_j,h_j)$ is gentle it follows that the vertical distance between $l_j$ and $h_j$ can be bounded:
\begin{align}~\label{eq:lem8_2d1}
    (\hat{y}_{h_j}-\hat{y}_{l_j})  \leq \frac{1}{A}(\hat{x}_{l_j}-\hat{x}_{h_j}) \iff A(\hat{y}_{h_j}-\hat{y}_{l_j}) + \hat{x}_{h_j}\leq \hat{x}_{l_j}.
\end{align}

Moreover, we can bound the length of the edge $(l_j,h_j)$ by
\begin{align}~\label{eq:lem8_2d2}
    d_2(l_j,h_j)\leq \sqrt{1+\frac{1}{A^2}+\frac{2|\cos(\theta)|}{A}}(\hat{x}_{l_j}-\hat{x}_{h_j}) .
\end{align}

Let $h_i,h_{i+1},...,h_{j-1}=h_j$ be the maximal low path ending at $h_j$. Note that by Lemma~\ref{lem:lem7} we have
\begin{align}~\label{eq:lem8_2d3}
    d_2^T(h_i,h_j)&\leq (\hat{x}_{h_j}-\hat{x}_{h_i})+(\hat{y}_{h_j}-\hat{y}_{h_i}) .
\end{align}

Next, note that either $h_i=h_0=a$ or $h_i$ is an eastern point in $\hat{P}_{i}$ that has a potential and Lemma~\ref{lem:lem5} applies. In this case we have that since $h_i$ is eastern in $\hat{P}_{i}$ we have
\begin{align}~\label{eq:lem8_2d4}
    \begin{split}
       d_2^T(a,h_i)
       &\leq (1+1/A)\hat{x}_{i}\\
       &= (1+1/A)\hat{x}_{h_i} .
    \end{split}
\end{align}

Using inequalities~\eqref{eq:lem8_2d1},~\eqref{eq:lem8_2d2},~\eqref{eq:lem8_2d3} and~\eqref{eq:lem8_2d4}, together with the triangle inequality, we get
\begin{align*}
     &\ d_2^T(a,l_j)-A\hat{y}_{l_j}\\
     \leq&\ d_2^T(a,h_i)+d_2^T(h_i,h_j)+d_2(h_j,l_j)-A\hat{y}_{l_j}\\
    \leq&\ (1+1/A)\hat{x}_{h_i} + (\hat{x}_{h_j}-\hat{x}_{h_i})+(\hat{y}_{h_j}-\hat{y}_{h_i})  +\sqrt{1+\frac{1}{A^2}+\frac{2|\cos(\theta)|}{A}}(\hat{x}_{l_j}-\hat{x}_{h_j})-A\hat{y}_{l_j}\\
    = &\ 1/A \hat{x}_{h_i} + \hat{x}_{h_j}+\sqrt{1+\frac{1}{A^2}+\frac{2|\cos(\theta)|}{A}}(\hat{x}_{l_j}-\hat{x}_{h_j})  +(\hat{y}_{h_j}-\hat{y}_{h_i})-A\hat{y}_{l_j}\\
    \leq&\ 1/A \hat{x}_{h_i} + \hat{x}_{h_j} +\sqrt{1+\frac{1}{A^2}+\frac{2|\cos(\theta)|}{A}}(\hat{x}_{l_j}-\hat{x}_{h_j}) + A(\hat{y}_{h_j}-\hat{y}_{h_i})-A\hat{y}_{l_j}\\
     &\ \phantom{.}\hskip 6.5cm \text{since   $A>1$}\\ 
     = &\ 1/A \hat{x}_{h_i} +\sqrt{1+\frac{1}{A^2}+\frac{2|\cos(\theta)|}{A}}\hat{x}_{l_j} +A(\hat{y}_{h_j}-\hat{y}_{l_j}) -A\hat{y}_{h_i}\\
     &\ -\left(\sqrt{1+\frac{1}{A^2}+\frac{2|\cos(\theta)|}{A}}-1\right)\hat{x}_{h_j}\\ \displaybreak
    \leq&\ \sqrt{1+\frac{1}{A^2}+\frac{2|\cos(\theta)|}{A}}\hat{x}_{l_j} + 1/A \hat{x}_{h_i} +A(\hat{y}_{h_j}-\hat{y}_{l_j})\\
    &\ \phantom{.}\hskip 3.5cm \text{since both  $-\hat{y}_{h_i}<0$ and $-\hat{x}_{h_j}<0$}\\ 
    \leq&\ \sqrt{1+\frac{1}{A^2}+\frac{2|\cos(\theta)|}{A}}\hat{x}_{l_j} + 1/A \hat{x}_{h_j} +A(\hat{y}_{h_j}-\hat{y}_{l_j})\\
    &\ \phantom{.}\hskip 3.5cm \text{since   $\hat{x}_{h_i}\leq \hat{x}_{h_j} $}\\
    \leq&\ \sqrt{1+\frac{1}{A^2}+\frac{2|\cos(\theta)|}{A}}\hat{x}_{l_j} +\hat{x}_{h_j} +A(\hat{y}_{h_j}-\hat{y}_{l_j})\\  &\ \phantom{.}\hskip 3.5cm \text{since   $A>1$}\\
     \leq&\ \sqrt{1+\frac{1}{A^2}+\frac{2|\cos(\theta)|}{A}}\hat{x}_{l_j} +\hat{x}_{l_j}\\ 
     &\ \phantom{.}\hskip 3.75cm \text{since  edge $(h_j,l_j )$ is gentle}\\
    \leq&\ \left(1+\sqrt{1+\frac{1}{A^2}+\frac{2|\cos(\theta)|}{A}}\right)\hat{x}_{l_j} .
\end{align*}
\end{enumerate}
\end{enumerate}
\end{proof}

 The following lemma bounds the length of a path which consists of points whose $\hat{y}$-coordinates differ greatly from $\hat{y}_b$. Since the path is monotone in the $\hat{x}$-direction, we need to bound the path from the point of view of $\hat{y}$-direction to attain the upper bound on the spanning ratio.
\begin{lemma}\label{lem:lem9}
 Suppose that no vertices of $\mathcal{P}$ are in $P(a,b)$ and that $(a,b)$ is not an edge. Let the coordinates of the inductive point $c$ of $\hat{P}_{i}$ be such that $0<L(\hat{x}_b-\hat{x}_c) < |\hat{y}_b-\hat{y}_c|$, for some $1 < i < k$.
    \begin{enumerate}
        \item[(i)] If $c=h_i$, and thus $0<L(\hat{x}_b-\hat{x}_c) < \hat{y}_c-\hat{y}_b$, then for the smallest $j > i$ such that $L(\hat{x}_b-\hat{x}_{h_j})\geq  \hat{y}_{h_j}-\hat{y}_b\geq 0$, all edges on the path $h_i,...,h_j$ are NE edges (refer to Figures~\ref{fig:lem9paraw}(a) and~\ref{fig:lem9parat}(a)).
        \item[(ii)] If $c=l_i$, and thus $0<L(\hat{x}_b-\hat{x}_c) < \hat{y}_b-\hat{y}_c$, then for the smallest $j > i$ such that $L(\hat{x}_b-\hat{x}_{l_j})\geq  \hat{y}_{b}-\hat{y}_{l_j}\geq 0$, all edges on the path $l_i,...,l_j$ are SE edges (refer to Figures~\ref{fig:lem9paraw}(b) and~\ref{fig:lem9parat}(b)).
    \end{enumerate}
\end{lemma}

\begin{figure}
%\centering
   \captionsetup{justification=centering}
   \centerline{
    \includegraphics[scale=1]{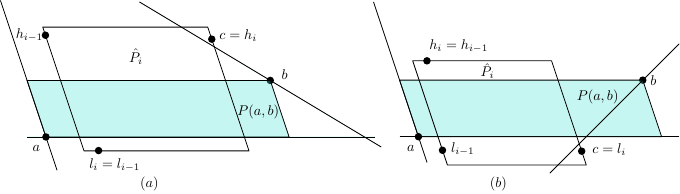}}
    \caption{Illustration of Lemma~\ref{lem:lem9} for Scenario 1. On the left $(a)$, the inductive point is $c=h_i$ and on the right $(b)$, the inductive point is $c=l_i$.}
    \label{fig:lem9paraw}
\end{figure}

\begin{figure}
%\centering
   \captionsetup{justification=centering}
   \centerline{
    \includegraphics[scale=1]{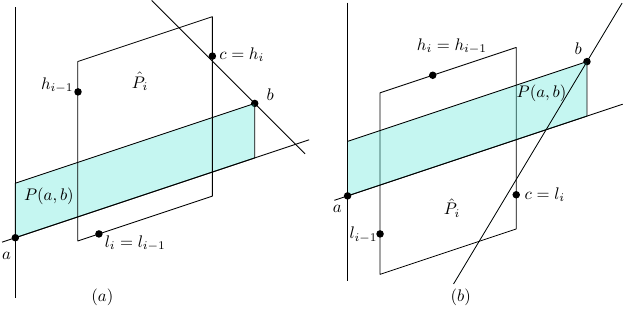}}
    \caption{Illustration of Lemma~\ref{lem:lem9} for Scenario 3. On the left $(a)$, the inductive point is $c=h_i$ and on the right $(b)$, the inductive point is $c=l_i$.}
    \label{fig:lem9parat}
\end{figure}

\begin{proof}
    We consider the case where $c=h_i$ as the case where $c=l_i$ can be shown using a similar argument. \\
    Notice that such a $j$ exists since choosing $h_j=b$ satisfies  $L(\hat{x}_b-\hat{x}_{h_j})\geq  \hat{y}_{h_j}-\hat{y}_b\geq 0$. Then for $i\leq k< j$, we have $0<L(\hat{x}_b-\hat{x}_{h_k}) < \hat{y}_{h_k}-\hat{y}_b$. Then since $\hat{x}_{h_{k+1}}> \hat{x}_{h_{k}}$ and neither vertex can be below the line $ab$, there are three options for the edge $(h_k, h_{k+1})$: WN, WE or NE in $\hat{P}_{k+1}$. Suppose for now that $h_k$ is on the W side of $\hat{P}_{k+1}$ , then we will show that $b$ must be inside $\hat{P}_{k+1}$, contradicting the fact that $\hat{P}_{k+1}$ must be empty.  By construction we know that $l_{k+1}$ must also be on the boundary of $\hat{P}_{k+1}$. Notice that the triangle formed by the WN, WS, SE corners of $\hat{P}_{k+1}$ can be expressed as the intersection of three half planes. We have that $b$ is on the right side of the half plane given by the W side of $\hat{P}_{k+1}$, since $\hat{x}_{h_k} < \hat{x}_b$. Next, we also have that $b$ is above the S side of $\hat{P}_{k+1}$, since $\hat{y}_{l_{k+1}}<\hat{y}_b$. Finally, by assumption we have that $\hat{y}_{h_k} > L(\hat{x}_b-\hat{x}_{h_k})+\hat{y}_b$ which is below the diagonal of $\hat{P}_{k+1}$ as $h_{k}$ is on the W side, hence $b$ is below the diagonal given by the WN and SE corners. Since $b$ is contained in the three half planes, then $b$ is in the interior of $\hat{P}_{k+1}$. As such, we have that $h_k$ cannot be on the W side and therefore, the edge $(h_k,h_{k+1})$ is NE. 
\end{proof}

Next we prove the main theorem of this section.
Recall that in all scenarios, the $\{\hat{x},\hat{y}\}$ basis was defined so that $L(\hat{x}_b-\hat{x}_a) >\hat{y}_b-\hat{y}_a $, i.e. $L\hat{x}_b\geq \hat{y}_b$.
Moreover, recall that at the beginning of the section,
we assumed, without loss of generality,
that the long side of $P$ is vertical
and the short side of $P$ has non-negative slope.

Let $\Delta_{\hat{x}}(a,b)$ (respectively $\Delta_{\hat{y}}(a,b)$) denote the $\hat{x}$-coordinate difference (respectively the $\hat{y}$-coordinate difference) between $a$ and $b$.

We now have all the ingredients needed to prove the main theorem.
\begin{theorem}~\label{theo:theo}
    Suppose $a$ and $b$ are vertices in the parallelogram Delaunay graph. When $A\,\Delta_{\hat{x}}(a,b)\geq \Delta_{\hat{y}}(a,b)$ (i.e. we are in Scenario 2 or 3), we have
    \begin{equation}~\label{eq:case1}
        d_2^T(a,b)\leq \left(A+\sqrt{1+A^2+2A|\cos(\theta)|}\right)\hat{x}_{b}+\hat{y}_b.
    \end{equation}
    Else, (i.e. we are in Scenario 1 or 4), we have 
    \begin{equation}~\label{eq:case2}
        d_2^T(a,b)\leq \left(1+\sqrt{1+\frac{1}{A^2}+\frac{2|\cos(\theta)|}{A}}\right)\hat{x}_b+A\hat{y}_b.
    \end{equation}
\end{theorem}

\begin{proof}
 To prove Theorem~\ref{theo:theo} we use induction on pairs of vertices $(a,b)$ ordered from closest to furthest in the following sense: the proximity of vertices $a$ and $b$ is associated with the size of the smallest scaled translate of $\hat{P}$ which contains $a$ and $b$ on its boundary. In this sense, the base case consists of pairs of vertices $a,b$ which lie on the boundary of the smallest empty scaled translate of $\hat{P}$. Note that there cannot be any other vertices of $\mathcal{P}$ in such a scaled translate, otherwise $a,b$ would not be the closest pair of points. Hence there is an edge between $a,b$. This satisfies the statement of the theorem for both cases. This concludes the base case.

 Let $a,b$ be two vertices in the parallelogram Delaunay graph. Now, suppose that the statement holds for any pair of vertices associated with a smaller parallelogram. We first investigate Case (1), where $P(a,b)$ is empty. Case (2), where $P(a,b)$ is not empty, will follow. 

    \begin{enumerate}
        \item[(1)]Suppose $P(a,b)$ is empty. We will analyse two subcases, (i) Scenarios 2 and 3, and then (ii) Scenarios 1 and 4.
        \begin{enumerate}
            \item[(i)] Since we are in Scenario 2 or 3, we have $L = A $ and $A\hat{x}_b\geq \hat{y}_b$.

         If the parallelogram Delaunay graph contains the edge $(a,b)$, then the statement holds. \\
            Next, assume that there is no inductive parallelogram in the sequence $\hat{P}_1,\hat{P}_2,..,\hat{P}_k$, then by case $(1)$ of Lemma~\ref{lem:crossing} we have $$d_2^T(a,b)\leq \left(A+\sqrt{1+A^2+2A|\cos(\theta)|}\right)\hat{x}_b+\hat{y}_b.$$ Now we study the case where there exists an inductive parallelogram. Assume that $\hat{P}_i$ is the first inductive parallelogram. As in Lemma~\ref{lem:crossing}, we have two subcases to consider. First we will look at the case where $h_i$ is the inductive point of $\hat{P}_i$, and the case where $l_i$ is the inductive point will follow. Suppose for now that $h_i$ is the inductive point of $\hat{P}_i$, then using case$(2)(a)$ of Lemma~\ref{lem:crossing}, we have that $$d_2^T(a,h_i)+(\hat{y}_{h_i}-\hat{y}_b) \leq \left(A+\sqrt{1+A^2+2A|\cos(\theta)|}\right)\hat{x}_{h_i}$$ which can be re-arranged to obtain a bound on the distance between $a$ and $h_i$:
            \begin{equation}~\label{eq:i11}
                d_2^T(a,h_i) \leq \left(A+\sqrt{1+A^2+2A|\cos(\theta)|}\right)\hat{x}_{h_i}-(\hat{y}_{h_i}-\hat{y}_b).
            \end{equation}As in Lemma~\ref{lem:lem9}, we suppose $j> i$ is minimal such that $A(\hat{x}_b-\hat{x}_{h_j})\geq \hat{y}_{h_j}-\hat{y}_{b}\geq 0$ and $h_j$ is on the E side of $\hat{P}_j$. Then since all edges on the path $h_i,...,h_j$ are NE edges, we can bound their lengths using the triangle inequality. We get that 
            \begin{equation}~\label{eq:i12}
                d_2^T(h_i,h_j)\leq (\hat{x}_{h_{j}}-\hat{x}_{h_{i}})+(\hat{y}_{h_{i}}-\hat{y}_{h_{j}}).
            \end{equation} Since the smallest scaled translate of $\hat{P}$ that has both $h_j$ and $b$ on its boundary is smaller than the scaled translate which has $a$ and $b$ on its boundary, we can use the induction hypothesis. Moreover we have $A(\hat{x}_b-\hat{x}_{h_j})\geq \hat{y}_{h_j}-\hat{y}_b\geq 0$, hence we use inequality~\ref{eq:case1}:
            \begin{equation}~\label{eq:i13}
                d_2^T(h_j,b) \leq \left(A+\sqrt{1+A^2+2A|\cos(\theta)|}\right)(\hat{x}_b-\hat{x}_{h_j}) +(\hat{y}_{h_j}-\hat{y}_b).
            \end{equation}
  Using inequalities~\ref{eq:i11},~\ref{eq:i12} and~\ref{eq:i13} as well as the triangle inequality we get that
            \begin{align*}
                d_2^T(a,b) \leq&\ d_2^T(a,h_i)+d_2^T(h_i,h_j)+d_2^T(h_j,b)\\
                 \leq&\ \left(A+\sqrt{1+A^2+2A|\cos(\theta)|}\right)\hat{x}_{h_i}-(\hat{y}_{h_i}-\hat{y}_b)\\
                 &\ +(\hat{x}_{h_{j}}-\hat{x}_{h_{i}})+(\hat{y}_{h_{i}}-\hat{y}_{h_{j}})\\
                &\ +\left(A+\sqrt{1+A^2+2A|\cos(\theta)|}\right)(\hat{x}_b-\hat{x}_{h_j})+(\hat{y}_{h_j}-\hat{y}_b)\\
                 =&\ \left(A+\sqrt{1+A^2+2A|\cos(\theta)|}\right)\hat{x}_{h_i}+(\hat{x}_{h_{j}}-\hat{x}_{h_{i}})\\
                 &\ +\left(A+\sqrt{1+A^2+2A|\cos(\theta)|}\right)(\hat{x}_b-\hat{x}_{h_j})\\
                 =&\ \left(A+\sqrt{1+A^2+2A|\cos(\theta)|}\right)\hat{x}_b\\
                 &\ +\left(A+\sqrt{1+A^2+2A|\cos(\theta)|}-1\right)(\hat{x}_{h_i}-\hat{x}_{h_j})\\
                 \leq&\ \left(A+\sqrt{1+A^2+2A|\cos(\theta)|}\right)\hat{x}_b\\
                 &\ \phantom{.}\hskip 5cm \text{since $\hat{x}_{h_i} \leq \hat{x}_{h_j}$.}\\
                 <&\ \left(A+\sqrt{1+A^2+2A|\cos(\theta)|}\right)\hat{x}_b+\hat{y}_b
            \end{align*}
            This completes the proof for the case where the inductive point of $\hat{P}_i$ is $h_i$. 

         We now focus on the case when $l_i$ is the inductive point of $\hat{P}_i$. Using case$(2)(b)$ of Lemma~\ref{lem:crossing}, we have that $$d_2^T(a,l_i)-\hat{y}_{l_i} \leq \left(A+\sqrt{1+A^2+2A|\cos(\theta)|}\right)\hat{x}_{l_i},$$ which can be re-arranged to obtain a bound on the distance between $a$ and $l_i$:
            \begin{equation}~\label{eq:i21}
                d_2^T(a,l_i) \leq \left(A+\sqrt{1+A^2+2A|\cos(\theta)|}\right)\hat{x}_{l_i}+\hat{y}_{l_i}.
            \end{equation}
            As in Lemma~\ref{lem:lem9}, we suppose $j > i$ is minimal such that $A(\hat{x}_b-\hat{x}_{l_j})\geq \hat{y}_{b}-\hat{y}_{l_j}\geq 0$ and $l_j$ is on the E side of $\hat{P}_j$. Then since all the edges on the path $l_i,...,l_j$ are SE edges we can bound their length using the triangle inequality. We get that 
            \begin{equation}~\label{eq:i22}
                d_2^T(l_i,l_j)\leq (\hat{x}_{l_j}-\hat{x}_{l_i})+(\hat{y}_{l_j}-\hat{y}_{l_i}).
            \end{equation} Since the smallest scaled translate of $\hat{P}$ that has both $l_j$ and $b$ on its boundary is smaller than the scaled translate which has $a$ and $b$ on its boundary, we can use the induction hypothesis. Moreover, since $A(\hat{x}_b-\hat{x}_{l_j})\geq \hat{y}_{b}-\hat{y}_{l_j}\geq 0$, hence we use inequality~\ref{eq:case1}:
            \begin{equation}~\label{eq:i23}
                d_2^T(l_j,b)\leq \left(A+\sqrt{1+A^2+2A|\cos(\theta)|}\right)(\hat{x}_b-\hat{x}_{l_j}) +(\hat{y}_{b}-\hat{y}_{l_j})
            \end{equation} Using inequalities~\ref{eq:i21},~\ref{eq:i22} and~\ref{eq:i23} as well as the triangle inequality we get
            \begin{align*}
                d_2^T(a,b)\leq &\ d_2^T(a,l_i)+d_2^T(l_i,l_j)+d_2^T(l_j,b)\\
                \leq &\  \left(A+\sqrt{1+A^2+2A|\cos(\theta)|}\right)\hat{x}_{l_i}+\hat{y}_{l_i}  +(\hat{x}_{l_j}-\hat{x}_{l_i})+(\hat{y}_{l_j}-\hat{y}_{l_i})\\
                &\ + \left(A+\sqrt{1+A^2+2A|\cos(\theta)|}\right)(\hat{x}_b-\hat{x}_{l_j})  +(\hat{y}_{b}-\hat{y}_{l_j})\\
                =&\ \left(A+\sqrt{1+A^2+2A|\cos(\theta)|}\right)\hat{x}_{l_i} +(\hat{x}_{l_j}-\hat{x}_{l_i})\\
                &\ + \left(A+\sqrt{1+A^2+2A|\cos(\theta)|}\right)(\hat{x}_b-\hat{x}_{l_j})+\hat{y}_{b}\\
                = &\  \left(A+\sqrt{1+A^2+2A|\cos(\theta)|}\right)\hat{x}_b+ \hat{y}_{b}\\
                &\ + \left(A+\sqrt{1+A^2+2A|\cos(\theta)|}-1\right)(\hat{x}_{l_i}-\hat{x}_{l_j})\\
                \leq &\  \left(A+\sqrt{1+A^2+2A|\cos(\theta)|}\right)\hat{x}_b+ \hat{y}_{b}\\
                &\ \phantom{.}\hskip 5cm \text{since $\hat{x}_{l_i} \leq \hat{x}_{l_j}$.}\\
            \end{align*}
            This completes the proof of Case (1)(i).

            \item[(ii)] Since we are in Scenario 1 or 4, we have $L = \frac{1}{A}$ and $\frac{1}{A}\hat{x}_b\geq\hat{y}_b$. 
       
        If the parallelogram Delaunay graph contains the edge $(a,b)$, then the statement holds.\\
       Next, assume that there are no inductive parallelograms in the sequence $\hat{P}_1,\hat{P}_2,..,\hat{P}_k$ then by case $(1)$ of Lemma~\ref{lem:crossing} we know that  \begin{align*}
            d_2^T(a,b)\leq \left(\frac{1}{A}+\sqrt{1+\frac{1}{A^2}+\frac{2|\cos(\theta)|}{A}}\right)\hat{x}_b+\hat{y}_b \leq \left(1+\sqrt{1+\frac{1}{A^2}+\frac{2|\cos(\theta)|}{A}}\right)\hat{x}_b+A\hat{y}_b
        \end{align*}
        Now we study the case where there exists an inductive parallelogram. Assume that $\hat{P}_i$ is the first inductive parallelogram. As in Lemma~\ref{lem:crossing} we have two subcases to consider. First we will look at the case where $h_i$ is the inductive point of $\hat{P}_i$, and the case where $l_i$ is the inductive point will follow. Suppose for now that $h_i$ is the inductive point of $\hat{P}_i$,  using case $(2)(c)$ of Lemma~\ref{lem:crossing} we have $$d_2^T(a,h_i)+A(\hat{y}_{h_i}-\hat{y}_b) \leq \left(1+\sqrt{1+\frac{1}{A^2}+\frac{2|\cos(\theta)|}{A}}\right)\hat{x}_{h_i},$$ which can be re-arranged to obtain a bound on the distance between $a$ and $h_i$:
        \begin{equation}\label{eq:ii11}
            d_2^T(a,h_i)\leq \left(1+\sqrt{1+\frac{1}{A^2}+\frac{2|\cos(\theta)|}{A}}\right)\hat{x}_{h_i}-A(\hat{y}_{h_i}-\hat{y}_b) .
        \end{equation} As in Lemma~\ref{lem:lem9}, we suppose $j>i$ is minimal such that $\frac{1}{A}(\hat{x}_b-\hat{x}_{h_j})\geq \hat{y}_{h_j}-\hat{y}_{b}\geq 0$ and $h_j$ is on the E side of $\hat{P}_j$. Then since all edges on the path $h_i,...,h_j$ are NE edges we can bound their length using the triangle inequality. We get that 
        \begin{equation}\label{eq:ii12}
            d_2^T(h_i,h_j)\leq (\hat{x}_{h_{j}}-\hat{x}_{h_{i}})+(\hat{y}_{h_{i}}-\hat{y}_{h_{j}}).
        \end{equation} Since the smallest scaled translate of $\hat{P}$ that has both $h_j$ and $b$ on its boundary is smaller than the scaled translate which has $a$ and $b$ on its boundary, we can use the induction hypothesis. Moreover we have $\frac{1}{A}(\hat{x}_b-\hat{x}_{h_j})\geq \hat{y}_{h_j}-\hat{y}_b\geq 0$, hence we use inequality~\ref{eq:case2}:
        \begin{equation}\label{eq:ii13}
             d_2^T(h_j,b)\leq  
            \left(1+\sqrt{1+\frac{1}{A^2}+\frac{2|\cos(\theta)|}{A}}\right)(\hat{x}_b-\hat{x}_{h_j}) +A(\hat{y}_{h_j}-\hat{y}_b).
        \end{equation}
       Using inequalites~\ref{eq:ii11},~\ref{eq:ii12} and~\ref{eq:ii13} as well as the triangle inequality we get
            \begin{align*}
                d_2^T(a,b) \leq&\ d_2^T(a,h_i)+d_2^T(h_i,h_j)+d_2^T(h_j,b)\\
                 \leq &\ \left(1+\sqrt{1+\frac{1}{A^2}+\frac{2|\cos(\theta)|}{A}}\right)\hat{x}_{h_i}-A(\hat{y}_{h_i}-\hat{y}_b) + (\hat{x}_{h_{j}}-\hat{x}_{h_{i}})+(\hat{y}_{h_{i}}-\hat{y}_{h_{j}})\\
                 &\ +\left(1+\sqrt{1+\frac{1}{A^2}+\frac{2|\cos(\theta)|}{A}}\right)(\hat{x}_b-\hat{x}_{h_j})  +A(\hat{y}_{h_j}-\hat{y}_b)\\
                 \leq&\ \left(1+\sqrt{1+\frac{1}{A^2}+\frac{2|\cos(\theta)|}{A}}\right)\hat{x}_{h_i}-A(\hat{y}_{h_i}-\hat{y}_b)  + (\hat{x}_{h_{j}}-\hat{x}_{h_{i}})+A(\hat{y}_{h_{i}}-\hat{y}_{h_{j}})\\
                 &\ + \left(1+\sqrt{1+\frac{1}{A^2}+\frac{2|\cos(\theta)|}{A}}\right)(\hat{x}_b-\hat{x}_{h_j})  +A(\hat{y}_{h_j}-\hat{y}_b)\\
                 &\ \phantom{.}\hskip 5cm \text{since $A \geq 1$,}\\ 
                 \leq&\ \left(1+\sqrt{1+\frac{1}{A^2}+\frac{2|\cos(\theta)|}{A}}\right)\hat{x}_b  + \sqrt{1+\frac{1}{A^2}+\frac{2|\cos(\theta)|}{A}}(\hat{x}_{h_i}- \hat{x}_{h_j})\\
                 \leq&\ \left(1+\sqrt{1+\frac{1}{A^2}+\frac{2|\cos(\theta)|}{A}}\right)\hat{x}_b \\
                 &\ \phantom{.}\hskip 5cm \text{since $\hat{x}_{h_i} \leq \hat{x}_{h_j}$.}\\ 
                 < &\ \left(1+\sqrt{1+\frac{1}{A^2}+\frac{2|\cos(\theta)|}{A}}\right)\hat{x}_b+A\hat{y}_b
            \end{align*}
           Now that we have proved the theorem when the inductive point of $\hat{P}_i$ is $h_i$, we can focus on the case where $l_i$ is the inductive point of $\hat{P}_i$. Using case $(2)(d)$ of Lemma~\ref{lem:crossing} we have that $$ d_2^T(a,l_i)-A\hat{y}_{l_i} \leq \left(1+\sqrt{1+\frac{1}{A^2}+\frac{2|\cos(\theta)|}{A}}\right)\hat{x}_{l_i},$$ which can be re-arranged to obtain a bound on the distance between $a$ and $l_i$
            \begin{equation}\label{eq:ii21}
                d_2^T(a,l_i)\leq \left(1+\sqrt{1+\frac{1}{A^2}+\frac{2|\cos(\theta)|}{A}}\right)\hat{x}_{l_i}+A\hat{y}_{l_i}.
            \end{equation} As in Lemma~\ref{lem:lem9}, we suppose $j > i$ is minimal such that $\frac{1}{A}(\hat{x}_b-\hat{x}_{l_j})\geq \hat{y}_{b}-\hat{y}_{l_j}\geq 0$ and $l_j$ is on the E side of $\hat{P}_j$. Then since all the edges on the path $l_i,...,l_j$ are SE edges we can bound their length using the triangle inequality. We get that
            \begin{equation}\label{eq:ii22}
                d_2^T(l_i,l_j)\leq (\hat{x}_{l_{j}}-\hat{x}_{l_{i}})+(\hat{y}_{l_{j}}-\hat{y}_{l_{i}}).
            \end{equation} Since the smallest scaled translate of $\hat{P}$ that has both $l_j$ and $b$ on its boundary is smaller than the scaled translate which has $a$ and $b$ on its boundary, we can use the induction hypothesis. Moreover, since $\frac{1}{A}(\hat{x}_b-\hat{x}_{l_j})\geq \hat{y}_{b}-\hat{y}_{l_j}\geq 0$, we use inequality~\ref{eq:case2}:
            \begin{equation}\label{eq:ii23}
                d_2^T(l_j,b)\leq
            \left(1+\sqrt{1+\frac{1}{A^2}+\frac{2|\cos(\theta)|}{A}}\right)(\hat{x}_b-\hat{x}_{l_j}) +A(\hat{y}_{b}-\hat{y}_{l_j}).
            \end{equation}
            Using inequalites~\ref{eq:ii21},~\ref{eq:ii22} and~\ref{eq:ii23} as well as the triangle inequality we get
            \begin{align*}
                d_2^T(a,b)\leq&\  d_2^T(a,l_i)+d_2^T(l_i,l_j)+d_2^T(l_j,b)\\
                \leq&\ \left(1+\sqrt{1+\frac{1}{A^2}+\frac{2|\cos(\theta)|}{A}}\right)\hat{x}_{l_i}+A\hat{y}_{l_i} + (\hat{x}_{l_{j}}-\hat{x}_{l_{i}})+(\hat{y}_{l_{j}}-\hat{y}_{l_{i}})\\
                &\ + \left(1+\sqrt{1+\frac{1}{A^2}+\frac{2|\cos(\theta)|}{A}}\right)(\hat{x}_b-\hat{x}_{l_j}) +A(\hat{y}_{b}-\hat{y}_{l_j})\\
                \leq&\ \left(1+\sqrt{1+\frac{1}{A^2}+\frac{2|\cos(\theta)|}{A}}\right)\hat{x}_{l_i}+A\hat{y}_{l_i} + (\hat{x}_{l_{j}}-\hat{x}_{l_{i}})+A(\hat{y}_{l_{j}}-\hat{y}_{l_{i}})\\
                &\ + \left(1+\sqrt{1+\frac{1}{A^2}+\frac{2|\cos(\theta)|}{A}}\right)(\hat{x}_b-\hat{x}_{l_j})  +A(\hat{y}_{b}-\hat{y}_{l_j})\\
                &\ \phantom{.}\hskip 5cm \text{since $A\geq 1$,}\\
                \leq&\ \left(1+\sqrt{1+\frac{1}{A^2}+\frac{2|\cos(\theta)|}{A}}\right)\hat{x}_b + \sqrt{1+\frac{1}{A^2}+\frac{2|\cos(\theta)|}{A}}(\hat{x}_{l_i}-\hat{x}_{l_j})\\
                &\ +A\hat{y}_{b} \\
                \leq&\ \left(1+\sqrt{1+\frac{1}{A^2}+\frac{2|\cos(\theta)|}{A}}\right)\hat{x}_b + A\hat{y}_{b}\\
                &\ \phantom{.}\hskip 5cm \text{since $\hat{x}_{l_i}\leq\hat{x}_{l_j}$.}\\
            \end{align*}
        This completes the proof for Case (1)(ii), which completes the proof for Case (1).
        \end{enumerate}
  
        \item[(2)] Assume $P(a,b)$ contains at least one vertex of $\mathcal{P}$. We distinguish (i) Scenarios 2 and 3 from (ii) Scenarios 1 and 4.        
        \begin{enumerate}
            \item[(i)] Assume we are in Scenario 2 or 3.

            \begin{figure}
           % \centering
   \captionsetup{justification=centering}
            \centerline{
            \includegraphics[scale=1]{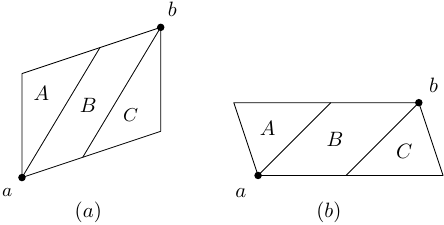}}
            \caption{Illustration of the three regions in $P(a,b)$. 
            Scenario 3 is shown on the left $(a)$, and Scenario 1 is shown on the right $(b)$.}
            \label{fig:theo_nempty}
            \end{figure}
            
             To simplify the proof, we partition $P(a,b)$ in three regions $\mathcal{A}$, $\mathcal{B}$ and $\mathcal{C}$. Intuitively, the regions are obtained by drawing the line through $a$ and the line through $b$ with the same slope as the positive diagonal of $\hat{P}$. The labelling of these regions $\mathcal{A}$, $\mathcal{B}$ and $\mathcal{C}$ is done from left to right (refer to Figure~\ref{fig:theo_nempty}$(a)$). \\
            More precisely, we have
            \begin{align*}
            \mathcal{A}=&\{ p \in P(a,b) : A(\hat{x}_p-\hat{x}_a)< \hat{y}_p-\hat{y}_a\}, \\
            \mathcal{B}=&\{ p \in P(a,b) : A(\hat{x}_p-\hat{x}_a) \geq \hat{y}_p-\hat{y}_a \land A(\hat{x}_b-\hat{x}_p) \geq \hat{y}_b-\hat{y}_p\},\\
            \mathcal{C}=&\{ p \in P(a,b) : A(\hat{x}_b-\hat{x}_p ) < \hat{y}_b-\hat{y}_p\}.
            \end{align*}
            
            Observe that a point $p\in\mathcal{A}$ also satisfies $A(\hat{x}_b-\hat{x}_p)> \hat{y}_b-\hat{y}_p $. Moreover a point $p\in\mathcal{C}$ also satisfies $A(\hat{x}_p-\hat{x}_a)<\hat{y}_p-\hat{y}_a$.
            
             When $p\in \mathcal{B}$, by definition we have that both $A(\hat{x}_p-\hat{x}_a)\geq \hat{y}_p-\hat{y}_a$ and $A(\hat{x}_b-\hat{x}_p)\geq \hat{y}_b-\hat{y}_p$ are satisfied. We can thus use inequality~\ref{eq:case1} on both pairs of vertices $(a,p)$ and $(p,b)$. 
            We get
            \begin{align*}
                d_2^T(a,b) \leq &\ d_2^T(a,p)+d_2^T(p,b)\\
                \leq &\ \left(A+\sqrt{1+A^2+2A|\cos(\theta)|}\right)(\hat{x}_p-\hat{x}_{a}) +(\hat{y}_{p}-\hat{y}_{a})\\
                &\ +\left(A+\sqrt{1+A^2+2A|\cos(\theta)|}\right)(\hat{x}_b-\hat{x}_{p}) 
                +(\hat{y}_{b}-\hat{y}_{p})\\
                = &\ \left(A+\sqrt{1+A^2+2A|\cos(\theta)|}\right)\hat{x}_p \\
                &\ +\left(A+\sqrt{1+A^2+2A|\cos(\theta)|}\right)(\hat{x}_b-\hat{x}_{p})+\hat{y}_{b}\\
                =&\  \left(A+\sqrt{1+A^2+2A|\cos(\theta)|}\right)\hat{x}_b+\hat{y}_{b} .
            \end{align*}

            Now we look at the case when there is no vertex in $\mathcal{B}$. Let $\hat{P}_a$ be the smallest homothet of $\hat{P}$ such that $a$ is on its W side and there exists a vertex $p\in \mathcal{A}$ on its boundary. In an analogous fashion, we let $\hat{P}_b$ be the smallest homothet of $\hat{P}$ such that $b$ is on its E side and there exists a vertex $q\in \mathcal{C}$ on its boundary. By assumption, there exists a vertex inside $P(a,b)$, thus either $p$ or $q$ exists. Assume for now that the vertex $p\in\mathcal{A}$ exists. The smallest scaled translate of $\hat{P}$ which has both $p$ and $b$ on its boundary is smaller than that of $a$ and $b$. Moreover we that $A(\hat{x}_b-\hat{x}_p ) > \hat{y}_b-\hat{y}_p$ hence we use inequality~\ref{eq:case1} to bound the length of the path between $p$ and $b$.  Therefore, we can apply the induction hypothesis on the pair $(p,b)$. When $(a,p)$ is an edge we have 
            \begin{align*}
               d_2^T(a,b)\leq &\ d_2^T(a,p)+d_2^T(p,b)\\ 
               = &\ d_2(a,p)+d_2^T(p,b)\\
               \leq&\  (\hat{x}_p-\hat{x}_{a})+(\hat{y}_{p}-\hat{y}_{a})+ \left(A+\sqrt{1+A^2+2A|\cos(\theta)|}\right)(\hat{x}_b-\hat{x}_{p}) +(\hat{y}_{b}-\hat{y}_{p})\\
               = &\ \hat{x}_p+ \left(A+\sqrt{1+A^2+2A|\cos(\theta)|}\right)(\hat{x}_b-\hat{x}_{p}) +\hat{y}_{b}\\
               \leq &\ \left(A+\sqrt{1+A^2+2A|\cos(\theta)|}\right)\hat{x}_b+\hat{y}_{b}
            \end{align*}
            The proof when a vertex $q$ exists and $(q,b)$ is an edge is similar.

             If we suppose now that $(a,p)$ is not an edge in the parallelogram Delaunay graph, then there exists at least one vertex in $\hat{P}_a \cap \mathcal{C}$. Let $p'\in \hat{P}_a \cap \mathcal{C}$ be the vertex closest to $a$, then we have that $(a,p')$ is an edge. The smallest scaled translate of $\hat{P}$ which has both $p'$ and $b$ on its boundary is smaller than that of $a$ and $b$. Moreover we have that $A(\hat{x}_b-\hat{x}_{p'} ) < \hat{y}_b-\hat{y}_{p'}$ hence we use inequality~\ref{eq:case2} to bound the length of the path between $p'$ and $b$.     
            \begin{align*}
                d_2^T(p',b)\leq &\ A\left(\hat{x}_b-\hat{x}_{p'}\right)  + \left(1+\sqrt{1+\frac{1}{A^2}+\frac{2|\cos(\theta)|}{A}}\right)(\hat{y}_{b}-\hat{y}_{p'}) .
            \end{align*}
            Since $p' \in \mathcal{C}$, we applied the induction hypothesis for instances of Scenarios 1 and 4 for the path from $p'$ to $b$. This explains why we swapped the $\{\hat{x},\hat{y}\}$ basis. We get
            \begin{align*}
                d_2^T(a,b)\leq &\ d_2^T(a,p')+d_2^T(p',b)\\
               =&\ d_2(a,p')+d_2^T(p',b)\\
               \leq &\ (\hat{x}_{p'}-\hat{x}_{a})+(\hat{y}_{p'}-\hat{y}_{a}) +  A(\hat{x}_b-\hat{x}_{p'})\\
               &\ +\left(1+\sqrt{1+\frac{1}{A^2}+\frac{2|\cos(\theta)|}{A}}\right)(\hat{y}_{b}-\hat{y}_{p'})\\
               = &\ \hat{x}_{p'}+\hat{y}_{p'} +  A(\hat{x}_b-\hat{x}_{p'})\\
               &\ +\left(1+\sqrt{1+\frac{1}{A^2}+\frac{2|\cos(\theta)|}{A}}\right)(\hat{y}_{b}-\hat{y}_{p'})\\
               \leq &\  A\hat{x}_b+\left(1+\sqrt{1+\frac{1}{A^2}+\frac{2|\cos(\theta)|}{A}}\right)\hat{y}_{b} \\
                &\ \phantom{.}\hskip 5cm \text{since $A\geq 1$.}\\
            \end{align*}
            Since we are in Scenario 2 or 3, $A\hat{x}_b\geq \hat{y}_b$ and we know that $1+\sqrt{1+\frac{1}{A^2}+\frac{2|\cos(\theta)|}{A}}>1$, we have\footnote{From the fact that if $n>m>0$ and $k>1$ then $kn+m>km+n$ (Proof: $kn+m>km+n \iff k(n-m)+(m-n)>0 \iff k(n-m)>(n-m)$ which is true for $k\geq 1$). In our case, we have that $n=A\hat{x}_b, m=\hat{y}_b$ and $k=1+\sqrt{1+\frac{1}{A^2}+\frac{2|\cos(\theta)|}{A}}$.}
            \begin{align*}
                d_2^T(a,b) &\leq A\hat{x}_b+\left(1+\sqrt{1+\frac{1}{A^2}+\frac{2|\cos(\theta)|}{A}}\right)\hat{y}_{b}\\
                &\leq \left(1+\sqrt{1+\frac{1}{A^2}+\frac{2|\cos(\theta)|}{A}}\right)A\hat{x}_{b}+\hat{y}_{b}\\ &=\left(A+\sqrt{1+A^2+2A|\cos(\theta)|}\right)\hat{x}_{b}+\hat{y}_{b}
            \end{align*}
            This completes the proof when we are in Scenario 2 or 3.

            \item[(ii)]
            In a similar way as in Scenario 2 or 3, we partition $P(a,b)$ in three regions $\mathcal{A}$, $\mathcal{B}$ and $\mathcal{C}$  (refer to Figure~\ref{fig:theo_nempty}$(b)$). 
            \begin{align*}
            \mathcal{A} =& \{p\in P(a,b) : A(\hat{x}_b-\hat{x}_p)\geq \hat{y}_b-\hat{y}_p \},\\
            \mathcal{B}=& \{p\in P(a,b) :  A(\hat{x}_b-\hat{x}_p)<\hat{y}_b-\hat{y}_p \land A(\hat{x}_p-\hat{x}_a) < \hat{y}_p-\hat{y}_a\}, \\
            \mathcal{C}=& \{p\in P(a,b): A(\hat{x}_p-\hat{x}_a)\geq \hat{y}_p-\hat{y}_a\}.
            \end{align*}
            
            Observe that a point $p\in\mathcal{A}$ also satisfies $A(\hat{x}_p-\hat{x}_a)>(\hat{y}_p-\hat{y}_a)$. Moreover a point $p\in\mathcal{C}$ also satisfies $A(\hat{x}_b-\hat{x}_p)<\hat{y}_b-\hat{y}_p$.
 
            When $p\in \mathcal{B}$, by induction, we have that both $A(\hat{x}_p-\hat{x}_a)< \hat{y}_p-\hat{y}_a$ and $A(\hat{x}_b-\hat{x}_p)< \hat{y}_b-\hat{y}_p$ are satisfied. We can thus use inequality~\ref{eq:case2} on both pairs of vertices $(a,p)$ and $(p,b)$ respectively. We get
            \begin{align*}
                d_2^T(a,b)\leq &\ d_2^T(a,p)+d_2^T(p,b)\\
                \leq &\  \left(1+\sqrt{1+\frac{1}{A^2}+\frac{2|\cos(\theta)|}{A}}\right)(\hat{x}_p-\hat{x}_{a}) +A(\hat{y}_{p}-\hat{y}_{a})\\
                &\ +\left(1+\sqrt{1+\frac{1}{A^2}+\frac{2|\cos(\theta)|}{A}}\right)(\hat{x}_b-\hat{x}_{p})  +A(\hat{y}_{b}-\hat{y}_{p})\\
                = &\ \left(1+\sqrt{1+\frac{1}{A^2}+\frac{2|\cos(\theta)|}{A}}\right)\hat{x}_p+A\hat{y}_{p}\\
                &\ +\left(1+\sqrt{1+\frac{1}{A^2}+\frac{2|\cos(\theta)|}{A}}\right)(\hat{x}_b-\hat{x}_{p})  +A(\hat{y}_{b}-\hat{y}_{p})\\
                =&\  \left(1+\sqrt{1+\frac{1}{A^2}+\frac{2|\cos(\theta)|}{A}}\right)\hat{x}_b+A\hat{y}_{b}.
            \end{align*}
             Now we look at the case when there is no vertex in $\mathcal{B}$. Let $\hat{P}_a$ be the smallest homothet of $\hat{P}$ such that $a$ is on its W side and there exists a vertex $p\in \mathcal{A}$ on its boundary. In an analogous fashion, we let $\hat{P}_b$ be the smallest homothet of $\hat{P}$ such that $b$ is on its E side and there exists a vertex $q\in \mathcal{C}$ on its boundary. By assumption, there exists a vertex inside $P(a,b)$, thus either $p$ or $q$ exists. Assume for now that the vertex $p\in\mathcal{A}$ exists. The smallest scaled translate of $\hat{P}$ which has both $p$ and $b$ on its boundary is smaller than that of $a$ and $b$. Moreover we that $A(\hat{x}_b-\hat{x}_p ) > \hat{y}_b-\hat{y}_p$ hence we use inequality~\ref{eq:case1} to bound the length of the path between $p$ and $b$.  
 Therefore, we can apply the induction hypothesis on the pair $(p,b)$. We get
            \begin{align*}
                d_2^T(p,b)\leq &\  \left(A+\sqrt{1+A^2+2A|\cos(\theta)|}\right)(\hat{y}_{b}-\hat{y}_{p})\\
                &\ +(\hat{x}_b-\hat{x}_{p}) .
            \end{align*}
            Since $p\in \mathcal{A}$, we applied the induction hypothesis from Scenarios 2 and 3 for the path from $p$ to $b$. This explains why we swapped the $\{\hat{x},\hat{y}\}$. 
            When $(a,p)$ is an edge we have:
            %In this case, we have that $A(\hat{x}_b-\hat{x}_p ) > \hat{y}_b-\hat{y}_p$ and the smallest homothet with $p$ and $b$ on its boundary is smaller than that of $a$ and $b$. If $(a,p)$ is an edge in the parallelogram Delaunay graph then notice that since $p\in A$, the induction hypothesis is as an instance of Scenario 2 or 3 for the path from $p$ to $b$ rather than Scenario 1 or 4 for the path from $a$ to $b$. This explains why we swap the $(\hat{x},\hat{y})$ basis in the following. 
            \begin{align*}
                d_2^T(a,b) \leq &\  d_2^T(a,p)+d_2^T(p,b)\\
                = &\ d_2(a,p)+d_2^T(p,b)\\
                \leq &\  (\hat{x}_p-\hat{x}_{a})+(\hat{y}_{p}-\hat{y}_{a})\\
                &\ + \left(A+\sqrt{1+A^2+2A|\cos(\theta)|}\right)(\hat{y}_{b}-\hat{y}_{p})  +(\hat{x}_b-\hat{x}_{p})\\
                =&\ \hat{y}_{p}+ \left(A+\sqrt{1+A^2+2A|\cos(\theta)|}\right)(\hat{y}_{b}-\hat{y}_{p}) + \hat{x}_b\\
                \leq &\ \left(A+\sqrt{1+A^2+2A|\cos(\theta)|}\right)\hat{y}_b+\hat{x}_{b} .
            \end{align*}
            Since $\frac{1}{A}\hat{x}_b>\hat{y}_b$, we have $\hat{x}_b>A\hat{y}_b$, and since\\
            $1+\sqrt{1+\frac{1}{A^2}+\frac{2|\cos(\theta)|}{A}}>1$, we find
            \begin{align*}
                d_2^T(a,b)&\leq \left(A+\sqrt{1+A^2+2A|\cos(\theta)|}\right)\hat{y}_b+\hat{x}_{b}\\
                &=\left(1+\sqrt{1+\frac{1}{A^2}+\frac{2|\cos(\theta)|}{A}}\right)A\hat{y}_b+\hat{x}_{b}\\
                &\leq \left(1+\sqrt{1+\frac{1}{A^2}+\frac{2|\cos(\theta)|}{A}}\right)\hat{x}_b+A\hat{y}_{b}. 
            \end{align*}
            The proof when a vertex $q$ exists and $(q,b)$ is an edge is similar.\\
            If we suppose now that $(a,p)$ is not an edge in the parallelogram Delaunay graph, then this means that there exists at least one vertex in $\hat{P}_a \cap \mathcal{C}$. Let $p'\in \hat{P}_a \cap \mathcal{C}$ be the vertex closest to $a$, then we have that $(a,p')$ is an edge. The smallest scaled translate of $\hat{P}$ which has both $p$ and $b$ on its boundary is smaller than that of $a$ and $b$. Moreover we that $A(\hat{x}_b-\hat{x}_{p'} ) < \hat{y}_b-\hat{y}_{p'}$ hence we use inequality~\ref{eq:case2} to bound the length of the path between $p$ and $b$. This gives us:
            \begin{align*}
                d_2^T(a,b)\leq &\ d_2^T(a,p')+d_2^T(p',b)\\
                = &\ d_2(a,p')+d_2^T(p',b)\\
                \leq &\ (\hat{x}_{p'}-\hat{x}_{a})+(\hat{y}_{p'}-\hat{y}_{a})\\
                &\ +\left(1+\sqrt{1+\frac{1}{A^2}+\frac{2|\cos(\theta)|}{A}}\right)(\hat{x}_b-\hat{x}_{p'}) +A(\hat{y}_{b}-\hat{y}_{p'})\\
                = &\ \hat{x}_{p'}+\hat{y}_{p'}+\left(1+\sqrt{1+\frac{1}{A^2}+\frac{2|\cos(\theta)|}{A}}\right)(\hat{x}_b-\hat{x}_{p'}) +A(\hat{y}_{b}-\hat{y}_{p'})\\
                \leq &\ \left(1+\sqrt{1+\frac{1}{A^2}+\frac{2|\cos(\theta)|}{A}}\right)\hat{x}_b+A\hat{y}_{b}
            \end{align*}
            This completes the argument for the proof of the theorem. 
        \end{enumerate}
    \end{enumerate}

\end{proof}

 Now that we have proved Theorem~\ref{theo:theo}, we can use the inequalities to bound the spanning ratio of the parallelogram Delaunay graph.\\
In the event that we are in Scenario 2 or 3, we get:
\begin{align}~\label{eq:parat_max}
\begin{split}
    \frac{d_2^T(a,b)}{d_2(a,b)} &\leq \frac{\left(A+\sqrt{1+A^2+2A|\cos(\theta)|}\right)\hat{x}_b+\hat{y}_b}{\sqrt{\hat{x}_b^2+\hat{y}_b^2-2\hat{x}_b\hat{y}_b\cos(\pi-\theta)}} = \frac{A+\sqrt{1+A^2+2A|\cos(\theta)|}+\frac{\hat{y}_b}{\hat{x}_b}}{\sqrt{1+(\frac{\hat{y}_b}{\hat{x}_b})^2+2\frac{\hat{y}_b}{\hat{x}_b}\cos(\theta)}}.
\end{split}
\end{align}
Let $f_{2,3}(\hat{y}_b/\hat{x}_b)$ be this function.

The derivative of $f_{2,3}(\hat{y}_b/\hat{x}_b)$ is equal to $0$ whenever
$$\frac{\hat{y}_b}{\hat{x}_b}=\frac{-\cos(\theta)\sqrt{1+A^2+2A|\cos(\theta)|}-A\cos(\theta)+1}{\sqrt{1+A^2+2A|\cos(\theta)|}+A-\cos(\theta)} ,$$
in which case $f_{2,3}(\hat{y}_b/\hat{x}_b)$ is equal to
\begin{align}
\label{max.span}
    \frac{\sqrt{2}\sqrt{1+A^2+A(|\cos(\theta)|-\cos(\theta))+(A-\cos(\theta))\sqrt{1+A^2+2A|\cos(\theta)|}}}{\sin(\theta)}.
\end{align}
Therefore,
the maximum of $f_{2,3}(\hat{y}_b/\hat{x}_b)$ is at least the expression in~\eqref{max.span}.

On the other hand, if we are in Scenario 1 or 4, we get that 
\begin{align}~\label{eq:paraw_max}
\begin{split}
    \frac{d_2^T(a,b)}{d_2(a,b)} &\leq  \frac{\bigg(1+\sqrt{1+\frac{1}{A^2}+\frac{2|\cos(\theta)|}{A}}\bigg)\hat{x}_b+A\hat{y}_b}{\sqrt{\hat{x}_b^2+\hat{y}_b^2-2\hat{x}_b\hat{y}_b\cos(\pi-\theta)}}= \frac{A+\sqrt{1+A^2+2A|\cos(\theta)|}+A^2\frac{\hat{y}_b}{\hat{x}_b}}{A\sqrt{1+(\frac{\hat{y}_b}{\hat{x}_b})^2+2\frac{\hat{y}_b}{\hat{x}_b}\cos(\theta)}}.
\end{split} 
\end{align}
Let $f_{1,4}(\hat{y}_b/\hat{x}_b)$ be this function.

The derivative of $f_{1,4}(\hat{y}_b/\hat{x}_b)$ is equal to $0$ whenever
$$\frac{\hat{y}_b}{\hat{x}_b}=\frac{-\cos(\theta)\sqrt{1+A^2+2A|\cos(\theta)|}+A^2-A\cos(\theta)}{\sqrt{1+A^2+2A|\cos(\theta)|}-A^2\cos(\theta)+A} ,$$
in which case $f_{1,4}(\hat{y}_b/\hat{x}_b)$ is equal to
\begin{align}
\label{candidate.max}
    \frac{\sqrt{(1+A^2)^2+2A(|\cos(\theta)|-A^2\cos(\theta))+2 A( 1- A \cos(\theta)) \sqrt{1 + A^2 + 2 A |\cos(\theta)|}}}{A\sin(\theta)}.
\end{align}
The three candidate values for the maximum of $f_{2,3}(\hat{y}_b/\hat{x}_b)$ (respectively $f_{1,4}(\hat{y}_b/\hat{x}_b)$) are 
\begin{enumerate}
    \item[(1)] $f_{2,3}(0)$ (respectively $f_{1,4}(0)$),
    \item[(2)] $\lim_{\hat{y}_b / \hat{x}_b \rightarrow \infty} f_{2,3}(\hat{y}_b/\hat{x}_b) = 1$ (respectively $\lim_{\hat{y}_b / \hat{x}_b \rightarrow \infty} f_{1,4}(\hat{y}_b/\hat{x}_b) = A$) and,
    \item[(3)]the value of $f_{2,3}$ (respectively $f_{1,4}$) when its derivative is $0$, which we denote by $f_{2,3}^*$ (respectively $f_{1,4}^*$).
\end{enumerate}
We can show that $f_{2,3}(0) \geq f_{1,4}(0)$,
$f_{2,3}^* \geq A$
and $f_{2,3}^* \geq f_{1,4}^*$.
Hence, the maximum occurs in Scenario 2 or 3. 
Moreover, we can show that the maximum value of $f_{2,3}(\hat{y}_b/\hat{x}_b)$ is obtained when its derivative is $0$.
As such, we find the following expression for the spanning ratio:
\begin{align*}
\frac{\sqrt{2}\sqrt{1+A^2+A(|\cos(\theta)|-\cos(\theta))+(A-\cos(\theta))\sqrt{1+A^2+2A|\cos(\theta)|}}}{\sin(\theta)}.
\end{align*}
The worst case of the above expression occurs when $\theta = \pi - \theta_0$, i.e., when we are in Scenario 2.
This allows us to conclude with the following.
\begin{theorem}
\label{thm.upper.bound}
The spanning ratio of a parallelogram Delaunay graph is at most
\begin{align}
\label{final.formula}
\frac{\sqrt{2}\sqrt{1+A^2+2A\cos(\theta_0)+(A+\cos(\theta_0))\sqrt{1+A^2+2A\cos(\theta_0)}}}{\sin(\theta_0)}.
\end{align}
\end{theorem}

\subsection{The Lower Bound}

\begin{theorem}
    There are parallelogram Delaunay graphs that have spanning ratio arbitrarily close to
    $$\frac{\sqrt{2}\sqrt{1+A^2+2A\cos(\theta_0)+(A+\cos(\theta_0))\sqrt{1+A^2+2A\cos(\theta_0)}}}{\sin(\theta_0)}$$
\end{theorem}
We now construct a set of points whose parallelogram Delaunay graph achieves the lower bound. We are in Scenario 2 with $A\geq 1$, we have $L=A$ and $\pi-\theta=\theta_0\leq \frac{\pi}{2}$.  We construct this point set as follows, let $a$ be the origin and let $b=\alpha\hat{x}+\beta\hat{y}$. We make two vertical columns of equidistant vertices $p_1,..,p_{n/2}$ and $q_1,.., q_{n/2}$ where $p_1=a$ and $p_{n/2}=(\beta+\alpha  A)\hat{y} $ and $q_1=b$ and $q_{n/2}=\alpha\hat{x}-\alpha  A\hat{y}$, refer to Figure~\ref{fig:lowerbound}$(a)$. Next, we move the vertices an arbitrary small distance in the horizontal direction, such that $p_i$ lies to the left of $p_{i+1}$ and $q_i$ lies to the left of $q_{i+1}$ for $1\leq i \leq n/2$, refer to Figure~\ref{fig:lowerbound}$(b)$. Moreover, we also show two triangles and their corresponding parallelograms. \\ 
\begin{figure}
%\centering
   \captionsetup{justification=centering}
   \centerline{
    \includegraphics[scale=0.822]{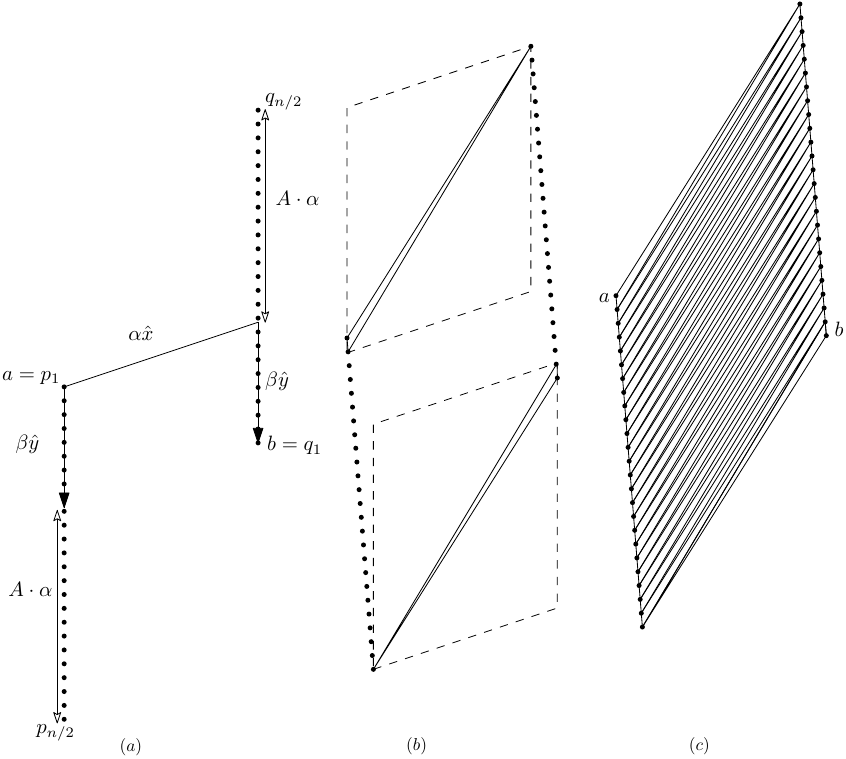}}
    \caption{Lower bound construction.}
    \label{fig:lowerbound}
\end{figure}
We proceed to analyze the length of one of the shortest paths between $a$ and $b$. Specifically the one via $p_{n/2}$, refer to Figure~\ref{fig:lowerbound}(c). Since all perturbations can be made arbitrarily small, this path has length
%\scalebox{0.825}{%
%$ (\beta+\alpha A)+\sqrt{\alpha^2+(-\alpha A)^2-2\alpha(-\alpha A)\cos(\pi-\theta)}= \alpha(A+\sqrt{1+A^2+2A\cos(\theta_0)})+\beta.$}\\
\begin{align*}
(\beta+\alpha A)+\sqrt{\alpha^2+(-\alpha A)^2-2\alpha(-\alpha A)\cos(\pi-\theta)}  = \alpha(A+\sqrt{1+A^2+2A\cos(\theta_0)})+\beta.
\end{align*}
The Euclidean distance between $a$ and $b$ is arbitrarily close to $$\sqrt{\alpha^2+\beta^2-2\alpha\beta\cos(\theta_0)}.$$ This implies that the spanning ratio is arbitrarily close to 
 \begin{align*}
\frac{\alpha(A+\sqrt{1+A^2+2A|\cos(\theta_0)|})+\beta}{\sqrt{\alpha^2+\beta^2-2\alpha\beta\cos(\theta_0)}} ,
 \end{align*}
which matches $f_{2,3}(\beta/\alpha)$ defined at the end of the previous section.

\section{Conclusion}
This paper generalized the approach by
Bonichon et al.~\cite{DBLP:journals/comgeo/BonichonGHP15}, and van Renssen et al.~\cite{DBLP:conf/esa/RenssenSSW23} to find the exact spanning ratio of the parallelogram Delaunay graph. This makes the parallelogram the fifth convex shape for which a tight spanning ratio is known. Note that the work of van Renssen et al. and of Bonichon et al. are special cases of this result. Indeed,
by setting $\theta_0=\pi/2$ in~\eqref{final.formula}, we obtain the exact spanning ratio for the rectangle Delaunay graph~\cite{DBLP:conf/esa/RenssenSSW23}.
If we also set $A = 1$,
we obtain the exact spanning ratio for the square Delaunay graph~\cite{DBLP:journals/comgeo/BonichonGHP15}.
Obviously, our bound also applies to diamonds.

Aside from the result itself, a significant technical contribution was the approach to generalizing the proof, which involves working with non-orthonormal bases. We believe this technique could be used to further broaden the class of convex shapes for which tight spanning ratios are known. For example, it may be feasible to apply this technique to the proof of Chew~\cite{Chew86} for the TD Delaunay graph. This would lead to an exact spanning ratio for the triangle Delaunay graph, where the triangle can be arbitrary.

\bibliography{ref}
\end{document}